\documentclass[12pt]{iopart}

\usepackage{iopams}
\usepackage{amsthm}
\usepackage{amssymb}
\usepackage{bbold}
\newtheorem{theorem}{Theorem}
\newtheorem{lemma}{Lemma}
\theoremstyle{definition}

\begin{document}
\small
\title{Perturbative renormalization of the lattice regularized $\phi_4^4$ with flow equations}
\author{Majdouline BORJI \footnote{majdouline.borji@polytechnique.edu},
Christoph KOPPER \footnote{christoph.kopper@polytechnique.edu}, }

\address{Centre de Physique Théorique CPHT, CNRS, UMR 7644}
\address{Institut Polytechnique de Paris, 91128 Palaiseau, France}
\vspace{10pt}
\begin{indented}
\item[]April 2020
\end{indented}

\begin{abstract}
The flow equations of the renormalization group allow to analyse the perturbative $n$-point functions of renormalizable quantum filed theories. Rigorous bounds implying renormalizability permit to control large momentum behaviour, infrared singularities and large order behaviour in the number of loops and the number of arguments $n$.\\
\indent In this paper, we analyse the Euclidean four-dimensional massive $\phi^4$ theory using lattice regularization. We present a rigorous proof that this quantum field theory is renormalizable, to all orders of the loop expansion based on the flow equations. The lattice regularization is known to break  Euclidean symmetry. Our main result is the proof of the restoration of the rotation and translation invariance in the renormalized theory   using the flow equations.
\end{abstract}

%
%
%
%
%


\section{Introduction}

Quantum field theory was originally developed as a theoretical framework that combines classical field theory, special relativity, and quantum mechanics and has become the general theoretical framework to study physical systems with an infinite (or large) number of degrees of freedom. 

A rigorous mathematical analysis of quantum field theories is faced with the problem that path integrals describing systems in field theory are generally not defined. There exists a complete theory of Gaussian measures that apply to free theories. However, for the interacting case, a rigorous mathematical description starts from regularized versions of the theory, where the number of degrees of freedom in space and momentum has been essentially made finite. This is common to all regularizations, such as momentum cutoff, Pauli-Villars regularization and lattice cutoff. One then studies correlation functions and proves that these have uniform limits in the cutoffs.

There are important situations in quantum field theory where perturbation theory does not produce quantitatively reliable results for the calculation of physical quantities. 
The most prominent example is the low-energy regime of Quantum Chromodynamics (QCD). So one would like to be able to analyse such theories nonperturbatively. By nonperturbative, we mean a method by which observables would directly be obtained to all orders in the coupling constant, without any expansion. 
One example, maybe the most important, is lattice field theory which consists in discretizing space-time. The continuous space-time is replaced by a discrete grid of points, the simplest arrangement being a hyper-cubic lattice. The distance between nearest neighbor sites is called the lattice spacing and usually denoted $a$. The inverse lattice spacing $a^{-1}$ provides a natural ultraviolet regularization.    
K. Wilson in 1974 \cite{11} introduced a formulation of Quantum Chromodynamics on a space-time lattice, which allows the application of various nonperturbative techniques. This discretization leads to a mathematically well-defined setting. Therefore, lattice field theory can be taken as a starting point for a mathematically clean approach to quantum field theory. It is a good starting point to derive properties of field theories in a rigorous
way.

 For finite lattice spacing $a$, the correlation functions are well approximated
for momenta well below the UV-cutoff $a^{-1}$. Renormalization
amounts to prove the existence of correlation functions in the continuum limit 
$a\rightarrow 0$ with certain properties. In this limit
a sequence of axioms must be satisfied in order to construct a Euclidean
quantum field theory. These are the well-known Osterwalder-Schrader axioms. In many important cases convergence is achieved by appropriately 
adjusting a finite number of parameters of the action and by a rescaling of the fields. These parameters (bare parameters) become functions of 
corresponding renormalized coupling constants. The renormalized coupling constants are defined by normalization conditions imposed on renormalized 
correlation functions at fixed Euclidean momenta. Renormalizability implies that all renormalized correlation functions, considered as functions of the 
renormalized parameters, stay well defined in the continuum limit, for all finite momenta $p$. The issue of renormalization theory is to show that a 
given field theory can be reparametrized in such a way that it stays finite if the UV-cutoff is removed and that the symmetries of the theory are preserved.

In perturbation theory,
the problem of renormalizability amounts to the study of Feynman integrals since the correlation functions are represented as a sum of Feynman integrals.
There exists a "power-counting theorem" that permits to determine the convergence of Feynman integrals in the large
cutoff limit by counting suitably defined UV-divergence degrees.
The Feynman integral associated with a Feynman diagram, together with
all its subintegrals, are required to have negative UV-divergence degrees so
that the Feynman integrals are (absolutely) convergent. Generically, Feynman
integrals are not convergent a priori. However, the UV-divergences can be subtracted order by order in perturbation theory, preserving locality. 
This provides a general renormalization prescription, for example, the BPHZ-subtraction scheme for continuum field theories
(the Bogoliubov-Parasiuk-Hepp-Zimmermann finite part prescription). It
applies to the integrand of momentum space Feynman integrals and does not
require  introducing an UV-cutoff. On the other hand, for a cutoff
theory the subtractions are arranged in such a way that they result from
local counterterms to the action. These counterterms provide the map between
bare and renormalized coupling constants and fields. They become
uniquely determined by imposing normalization conditions on the Green
or Schwinger functions.

 A further issue is to prove that a theory showing a symmetry can
be renormalized in such a way that the symmetry is preserved. This is
highly nontrivial for theories which are symmetric under a
nonlinear and/or local symmetry transformation, as in particular Yang-Mills theories like QCD and the electroweak sector of the standard model.

 Renormalization theory can also be studied directly in the framework of the Wilson renormalization group \cite{9,10}. In this framework the theories are described by an effective action $L^{\Lambda,\Lambda_0}$, depending on a scale $\Lambda$ with $0 \leq \Lambda \leq \Lambda_0 < \infty$ for Euclidean quantum field theories in the continuum with a momentum cutoff. Here $\Lambda$ plays a similar role as an infrared cutoff, $\Lambda_0$ denotes the ultraviolet cutoff. 
 $L^{\Lambda,\Lambda_0}$ should satisfy the following conditions:\\
$\bullet$ At the ultraviolet cutoff $\Lambda=\Lambda_0$, $L^{\Lambda,\Lambda_0}$ coincides with the bare action.\\
$\bullet$ For $\Lambda<\Lambda_0$, $L^{\Lambda,\Lambda_0}$ is obtained upon integration of the field degrees of freedom which propagate with momenta $p$ roughly between $\Lambda$ and $\Lambda_0$.\\
$\bullet$ As $\Lambda \rightarrow 0$, $L^{\Lambda,\Lambda_0}$ approaches the effective action, i.e. the generating functional of the (connected amputated)
Schwinger functions, of a theory without infrared cutoff. Thus the final effective action contains the full information of the original action that evolves under a change of scale.
Changing the infrared cutoff $\Lambda$ leads to renormalization group equations
which describe the scale dependence of the effective theories on $\Lambda$ in a
compact way. When $\Lambda$ varies continuously, the resulting flow equations
are first-order differential equations in the infrared cutoff $\Lambda$. Solving them 
under appropriate boundary conditions (at $\Lambda=0$ or $\Lambda=\Lambda_0$) amounts to determine
the infrared and ultraviolet properties of a field theory.

Polchinski and later Keller, Kopper and Salmhofer showed that these ideas also lead to a simplified proof of perturbative renormalizability of quantum field theories \cite{5,7}. Usually, complete proofs of
renormalizability are rather cumbersome, because of the complex combinatorics of overlapping ultraviolet divergences of a Feynman diagrammatic approach. They require a power-counting theorem which ensures finiteness of multi-dimensional Feynman integrals by imposing the appropriate
subtractions. In the framework of flow equations this complicated analysis
is avoided. It gives an alternative proof based on a tight inductive scheme wherefrom bounds on the regularized correlation functions implying renormalizability can be deduced. Renormalizability of a quantum field theory implies that the unregularized correlation functions
$$\lim_{\Lambda \rightarrow 0, \Lambda_0 \rightarrow \infty} \mathcal{L}^{\Lambda,\Lambda_0}_{l,n}(p_1,...,p_n)$$
exist in the sense that they are both IR (in massless theories outside exceptional momentum configurations)
and UV finite. Finite limits are achieved by imposing a finite set  of renormalization conditions on a physical scale 
that is independent of the UV cutoff $\Lambda_0$.  We will consider the case in which all fields are massive to avoid IR problems. 
Proving renormalizability then basically amounts to show the existence of the large UV-cutoff limit $\Lambda_0\rightarrow \infty$.

 In the present work, we investigate the renormalizability of massive $\phi_4^4$-theory regularized by a lattice cut-off. 
 The proof of  perturbative renormalizability of a lattice regularized field theory is not direct from the usual power counting theorems. 
 The well known power counting theorems of Weinberg \cite{16}, and Hahn, Zimmermann \cite{17} which state sufficient conditions for the convergence of Feynman integrals do not apply in the presence of a lattice cutoff. Reisz \cite{13} has given a generalization of the power counting theorem for a wide class of lattice field theories where a new kind of an ultraviolet divergence degree is used. The existence of a power counting theorem ensures that the combinatorics of subtractions to renormalize a diagram is described by Zimmermann's forest formula \cite{12}. The situation is different for a lattice field theory. Reisz \cite{15} has proved that the counterterms instead of being polynomials are periodic functions in the external momenta, which can be obtained with the help of new operators he introduced, called subtraction operators. 
 
 The renormalization of lattice regularized $\phi^4_4$ theory in Polchinski's framework has been adressed in \cite{19}. 
The paper presents interesting arguments, but it does not aim at mathematical rigour and thus leaves certain mathematical questions 
unsolved, in particular w.r.t. to $O(4)$ and translation invariance of the continuum limit.

 Davoudi and Savage \cite{20} proposed a mechanism for the restoration of rotational symmetry in the continuum limit of lattice field theories on hyper cubic lattices. The approach is based on constructing smeared lattice operators that smoothly evolve into continuum operators with definite angular momentum as the lattice-spacing is reduced. However, this method regards only finite lattices and the full recovery of rotational invariance in the lattice theories requires the suppression of rotational symmetry breaking contributions to the physical quantities not only as a result of short-distance discretization effects,
but also as a result of boundary effects of the finite cubic lattice. More precisely, the rotational invariant theory is achieved as the lattice becomes infinitely large, corresponding to an infinitely large number of points in momentum space.  Here we give a proof of rotation symmetry restoration for $\phi^4_4$ lattice 
regularized field theory on an infinite lattice.

The paper is organized as follows: In section \ref{Sec2}, we introduce the flow equations. In section \ref{Sec3} we present the steps of proving renormalizability of four-dimensional $\phi^4$ theory on the lattice by means of the flow equations, following \cite{6}. Renormalizability is stated in terms of uniform bounds on the (coefficient functions of the) solution $L^{a_0,a}(\phi)$ of the flow equation and its derivative with respect to the lattice cutoff $a^{-1}$, with boundary conditions 
imposed at $a=\infty$ for the relevant couplings and at $a=a_0$ for the irrelevant interactions. 

 Sections \ref{Sec4} and \ref{Sec5} are at the heart of this paper. In section \ref{Sec4} we introduce the rotated lattice and we show that the 
 differences $\mathcal{D}_{l,n}^{a_0,a,O}(p_1,\cdots,p_n)$ of the correlation functions of arguments defined on the rotated lattice and on 
 the original  lattice:
$$\mathcal{D}_{l,n}^{a_0,a,O}(p_1,\cdots,p_n):= \mathcal{L}^{a_0,a,O}_{l,n}(Op_1,\cdots,Op_n)-\mathcal{L}^{a_0,a}_{l,n}(p_1,\cdots,p_n)$$
converge to zero when $a_0\rightarrow 0$ and $a\rightarrow \infty$. In section \ref{Sec5} we give a proof of the existence of the continuum limit in 
 position space in the sense of tempered distributions. We find that the obtained limit is invariant under translations which concludes the 
 restoration of the Euclidean symmetries in the continuum limit.  


\section{The flow equations}\label{Sec2}
We consider  $\phi^4$ scalar field theory on four dimensional Euclidean space. We will formulate our theory with a lattice cutoff 
in the standard path integral formalism, where the lattice refers to the discretization of space-time. In the following, we introduce 
general notions of a space-time lattice and the $\phi^4$  model on the lattice, but only to the extent that is relevant to this paper.  


\subsection{Lattice field theory}

The four-dimensional hypercubic lattice is a set of sites denoted by 
$$\Lambda_{a_0}=a_0\mathbb{Z}^4$$
where $a_0$ denotes the lattice spacing in Euclidean time and spatial directions. 

 One of the first questions in lattice field theory is how to put a model on
the lattice once it is defined on the space-time continuum. The question
refers both to the framework of classical field theory, i.e. at the level of the
classical action, and to quantum field theory. Naturally discretization of space and time 
implies that differentiation with respect to space and time is to be replaced by a 
corresponding difference operation. 


\subsection{$\phi^4$ scalar field theory on the lattice}

Perturbative renormalizability of euclidean $\phi^4_4$ theory will be established by analysing the generating functional 
$L^{a_0,a}$ of connected (free propagator) amputated Schwinger functions (CAS). The upper indices $a_0$ and $a$ 
enter through the regularized propagator 
\begin{equation}\label{propagator}
    C^{a_0,a}(p)=\frac{1}{\hat{p}^2+m^2}\left(e^{-a_0^2(\hat{p}^2+m^2)}-e^{-a^2(\hat{p}^2+m^2)}\right)
\end{equation}{}
where the map $\hat{p}:=\left(\hat{p}(p_{\mu})\right)_{1\leq \mu \leq 4}$ is defined as follows
\begin{equation}{\label{lattMom}}
\begin{array}{ccccc}
\hat{p} & : & {\left]-\frac{\pi}{a_0},\frac{\pi}{a_0}\right[} & \to & {\left]-\frac{2}{a_0},\frac{2}{a_0}\right[} \\ \\
 & & p_{\mu} & \mapsto & \frac{2}{a_0}\sin(\frac{a_0 p_{\mu}}{2})\\
\end{array}
\end{equation}
In the sequel we shall write with slight abuse of notation $$C^{a_0,a}(\hat{p}):=C^{a_0,a}(p),\qquad \hat{p}(p_{\mu}):=\hat{p}_{\mu}$$
Upon removal of the cutoffs, i.e. in the limit $a_0\rightarrow 0$, $a\rightarrow \infty$, we indeed recover the free propagator 
$\frac{1}{p^2+m^2}$. For the Fourier transform we use the convention
\begin{equation}
    \hat{f}(x)=\int_{p,\mathcal{B}_{a_0}}f(p)e^{ip\cdot x}:=\int_{\left]-\frac{\pi}{a_0},\frac{\pi}{a_0}\right[^4}\frac{d^4p}{(2\pi)^4}f(p)e^{ip\cdot x}
\end{equation}
using~the~shorthand
\[
\int_{p,\mathcal{B}_{a_0}}:= \int_{\left]-\frac{\pi}{a_0},\frac{\pi}{a_0}\right[^4}\frac{d^4p}{(2\pi)^4}
\  \ \mbox{  with }\  \ \mathcal{B}_{a_0}=\left]-\frac{\pi}{a_0},\frac{\pi}{a_0}\right[^4
\] 
denoting the first Brillouin zone.
For the inverse Fourier transform we write
\begin{equation}
    f(p)=a_0^4\sum_{x \in \Lambda_{a_0}}\hat{f}(x)e^{-ip\cdot x}
\end{equation}
so that in position space
$$\hat{C}^{a_0,a}(x,y)=\int_{p,\mathcal{B}_{a_0}}C^{a_0,a}(\hat{p})e^{ip\cdot (x-y)} $$
We assume 
$$0\leq a_0 \leq a \leq \infty$$
so that the Wilson flow parameter $1/a$ takes the role of an IR cutoff, whereas $1/a_0$ is the UV cutoff. We introduce the convention
$$\hat{\phi}_{a_0}(x)=\int_{p,\mathcal{B}_{a_0}}\phi_{a_0}(p)e^{ip\cdot x},~~~~~\frac{\delta}{\delta \hat{\phi}_{a_0}(x)}=\int_{p,\mathcal{B}_{a_0}} 
\frac{\delta}{\delta \phi_{a_0}(p)}e^{-ip\cdot x}$$
For our purposes the field $\hat{\phi}_{a_0}(x)$ may be assumed to live in the Hilbert space $l_2\left(\Lambda_{a_0}\right)$ endowed with the inner scalar product
$$
\langle f,g \rangle_{l_2\left(\Lambda_{a_0}\right)}=a_0^4\sum_{x \in \Lambda_{a_0}}f(x)\overline{g(x)}
$$
Our starting point  is the  bare action of symmetric $\phi_4^4$ theory
 \begin{eqnarray}\label{bareinter}
 \fl L^{a_0,a_0}(\hat{\phi}_{a_0})=a_0^4\sum_{x \in \Lambda_{a_0}} \left \{ \frac{f}{4!}\hat{\phi}_{a_0}^4+d(a_0)\hat{\phi}_{a_0}^2+b(a_0)(\hat{\partial}_{\mu,a_0}\hat{\phi}_{a_0})^2+c(a_0)\hat{\phi}^4_{a_0}\right \}
  \end{eqnarray}
$$ d(a_0), c(a_0)=\mathcal{O}(\hbar)\ , \quad
b(a_0)=\mathcal{O}(\hbar^2)$$
The differentiation in (\ref{bareinter}) is defined by the difference operator
$$
\left( \hat{\partial}_{\mu}\hat{\phi}_{a_0}\right)(x)=\frac{\hat{\phi}_{a_0}(x+a_0 e_{\mu})-\hat{\phi}_{a_0}(x)}{a_0} 
$$
for $x \in \Lambda_{a_0}$, $e_{\mu}$ is the unit vector in the $\mu^{\mathrm{th}}$ coordinate direction. The first term is formed of the field's self-interaction with real coupling constant $f$ having mass dimension equal to zero. The second part contains the related counter terms, determined according to the following rule. The canonical mass dimension of the field is one, the counter terms allowed in the bare interaction are all local terms of mass dimension $\leq$ 4 formed out of the field and its derivatives respecting cubic lattice symmetry. The $O(4)$ and translation symmetries are violated by the lattice regularization.
From the bare action and the flowing propagator, we may define Wilson's flowing effective action $L^{a_0,a}$ by integrating out momenta 
roughly in the region $1/a^2\leq p^2 \leq 1/a_0^2$. It is defined through 
\begin{eqnarray}\label{fl}
 e^{-\frac{1}{\hbar}\left(L^{a_{_0},a}(\hat{\phi}_{a_{_0}})+I^{a_{_0},a}\right)}:&=\int d\mu_{a_{_0},a}(\Phi)  e^{-\frac{1}{\hbar}L^{a_{_0},a_{_0}}(\Phi+\hat{\phi}_{a_{_0}})}
 \ ,\quad
     L^{a_0,a}(0)&=0
\end{eqnarray}
and can be recognized to be the generating functional of the CAS of the theory with propagator $\hat{C}^{a_0,a}$ 
 and bare action $L^{a_0,a_0}$. In (\ref{fl}), $d\mu_{a_0,a}(\Phi)$ denotes the Gaussian measure with covariance 
 $\hbar \hat{C}^{a_0,a}$. It is proved in [1] that such a measure exists as a lattice approximation of the continuum gaussian measure. 
 ${I}^{a_0,a}$ denotes the field independent so called vacuum contributions. It is finite only in the finite volume approximation. 
 The infinite volume limit is taken only when it has been eliminated \cite{6}. We do not make the finite volume explicit here since it plays no role in the sequel.\\
The fundamental tool for our study of the renormalization problem is the functional flow equation
\begin{eqnarray}\label{FE}
\fl    \partial_{1/a}L^{a_0,a}=\frac{\hbar}{2}\langle \frac{\delta}{\delta \hat{\phi}_{a_0}},\left(\partial_{1/a}\hat{C}^{a_0,a}\right)*\frac{\delta}{\delta \hat{\phi}_{a_0}}\rangle L^{a_0,a} -\frac{1}{2} \langle \frac{\delta L^{a_0,a}}{\delta \hat{\phi}_{a_0}}, \left(\partial_{1/a}\hat{C}^{a_0,a}\right) *\frac{\delta L^{a_0,a}}{\delta \hat{\phi}_{a_0}} \rangle
\end{eqnarray}
By $\langle \cdot , \cdot \rangle$ we denote the scalar product in $l_2\left(\Lambda_{a_0}\right)$. (\ref{FE}) is obtained by deriving both sides of the equation (\ref{fl}) with respect to $1/a$ and performing an integration by parts in the functional integral on the RHS using the properties of the lattice Gaussian measure \cite{1}, and finally rearranging the powers of $\hbar$ coming from $L^{a_0,a}/ {\hbar}$ and from $\hbar \partial_{1/a}\hat{C}^{a_0,a}$ \cite{6}. To derive the flow equations verified by the $n$-point correlation functions, we first expand $L^{a_0,a}$ in moments for all $(p_i)_{1 \leq i \leq n} \in \mathcal{B}_{a_0}$ with respect to $\phi_{a_0}$, 
\begin{eqnarray*}
\fl ~~~~(2\pi)^{4(n-1)} \delta_{\phi_{a_0}(p_1)}\cdots\delta_{\phi_{a_0}(p_n)}L^{a_0,a}|_{\phi_{a_0}=0}=\delta^{4}_{\left[\frac{2\pi}{a_0}\right]}(p_1+\cdots+p_n)\mathcal{L}_n^{a_0,a}(p_1,\cdots,p_n)
\end{eqnarray*}
where we have written $\delta_{\phi_{a_0}(p)}=\delta/\delta \phi_{a_0}(p)$ and $\delta^{4}_{\left[\frac{2\pi}{a_0}\right]}:=\sum_{k\in \mathbb{Z}^4}\delta^{(4)}_{\frac{2k\pi}{a_0}}$.  We also expand in a formal powers series with respect to $\hbar$ to select the loop order $l$,
$$\mathcal{L}_n^{a_0,a}=\sum_{l=0}^{\infty}\hbar^{l}\mathcal{L}_{l,n}^{a_0,a}$$
From the functional flow equation (\ref{FE}), we then obtain the perturbative flow equations for the (connected free propagator amputated)
 $n$-point functions by identifying coefficients
\begin{equation}
\label{feq}
\!\!\! \!\!\! \!\!\!  \!\!\! \!\!\!  \partial_{1/a} \partial^w \mathcal{L}_{l,n}^{a_0,a}(p_1,\cdots,p_n)=
\frac{1}{2} \int_{k,\mathcal{B}_{a_0}} \partial^w \mathcal{L}_{l-1,n+2}^{a_0,a}(k,p_1,\cdots,p_n,-k)
\partial_{1/a}C^{a_0,a}(\hat k)
\end{equation}
$$
\! -\frac{1}{2} \!\sum \limits_{l_1,l_2}^{'}\! \sum_{n_1,n_2}^{'} \!\! \sum^{'}_{w_i}\!c_{w_i}\!\!
\left[\partial^{w_1}
\mathcal{L}_{l_1,n_1+1}^{a_0,a}(p_1,\cdots,p_{n_1},p)\partial^{w_3}\partial_{1/a}C^{a_0,a}(\hat{p})
\partial^{w_2}\mathcal{L}_{l_2,n_2+1}^{a_0,a}(-p,p_{n_1+1},\cdots,p_n)\right]_{rsy}
$$
$$p\equiv-p_1-\cdots-p_{n_1}\equiv p_{n_1+1}+\cdots+p_n~\left[\frac{2\pi}{a_0}\right]$$
Here we wrote (\ref{feq}) directly in a form where a number $|w|$ of momentum derivatives, characterized by a multi index
$w$, act on both sides, and we used the shorthand notation
\begin{eqnarray}
   \fl \partial^{w}:=\prod^{n}_{i=1}\prod_{\mu=0}^3\left(\frac{\partial}{\partial p_{i,\mu}}\right)^{w_{i,\mu}}~ 
   \mathrm{with}~ w=(w_{1,0},\cdots,w_{n,3}),~|w|=\sum{w_{i,\mu}},~w_{i,\mu}\in \mathbb{N}^*
\end{eqnarray}{}
The symbol $"rsy"$ means summation over those permutations of the momenta $p_1,\cdots,p_n$, which do not leave invariant the (unordered) 
subsets $(p_1,\cdots,p_{n_1})$ and $(p_{n_1+1},\cdots,p_n)$, and therefore, produce mutually different pairs of (unordered) image subsets, and 
the primes restrict the summations to $n_1+n_2=n$, $l_1+l_2=l$, $w_1+w_2+w_3=w$, respectively. Moreover, the combinatorial factor 
$c_{\left \{w_i\right \}}=w!(w_1!w_2!w_3!)^{-1}$ comes from Leibniz's rule. In the loop order $l=0$, the first term on the RHS is absent.


\section{Renormalization of lattice $\phi_4^4$ theory}\label{Sec3}

Perturbative renormalizability of the regularized field theory (\ref{fl}) amounts to the following: For given coupling constant $f$ in the 
bare interaction (\ref{bareinter}), the coefficients $d(a_0)$, $b(a_0)$ and $c(a_0)$ of the counter-terms can be adjusted within a loop 
expansion of the theory,
 $$d(a_0)=\sum_{l=1}^{\infty}\hbar^l d_l(a_0), \quad   b(a_0)=\sum_{l=2}^{\infty}\hbar^l b_l(a_0) ,\quad  c(a_0)=\sum_{l=1}^{\infty}\hbar^l c_l(a_0)$$
in such a way that the limits of the lattice $n-$point CAS functions exist when $a_0$ goes to $0$ and $a$ goes to $\infty$ in every loop order $l$.
$$ \forall \left(p_i\right)_{1\leq i \leq n} \in \mathbb{R}^4, \exists \tilde{a}_0>0 \  \mbox{ such that uniformly in }\ \mathcal{B}_{{\tilde a}_0}\,:$$
\begin{equation}
 \mathcal{L}_{l,n}^{0,\infty}(p_1,\cdots,p_n):=\lim_{a_0\rightarrow 0, a_0 \leq \tilde{a}_0}
 \lim_{a \rightarrow \infty}\mathcal{L}_{l,n}^{a_0,a}(p_1,\cdots,p_n),~n \in \mathbb{N}, l\in \mathbb{N}^*
 \label{limit1}
\end{equation}
The parameter $\tilde{a}_0$ guarantees that $\left(p_i\right)_{1\leq i \leq n}\in \mathcal{B}_{\tilde {a}_0}\subset \mathcal{B}_{a_0}$ for all $a_0\leq \tilde{a}_0$ so that they are well defined as arguments of the regularized $n$-point functions $\mathcal{L}^{a_0,a}_{l,n}$. 
The lattice breaks Euclidean symmetry and an essential point to the renormalizability of the theory is to prove the restoration of this symmetry. 
We will analyse the limits  $\mathcal{L}_{l,n}^{0,\infty}(p_1,\cdots,p_n)$ and prove  in particular their invariance 
under rotations and translations in sections \ref{Sec4} and \ref{Sec5}.


\subsection{Propagator bounds}\label{props}

The subsequent bounds on the CAS functions will depend heavily on the propagator of the theory we consider.
The bare propagator is, apart from the renormalization conditions, the main ingredient which decides what kind
of  bounds can be achieved. In this subsection we  collect the bounds on the propagator and its derivatives
we will need subsequently. From the definition (\ref{propagator}) we directly obtain
 \begin{equation}
 \label{propaa}
 \partial_{1/a}C^{a_0,a}(\hat{p})=(-2a^3)e^{-a^2(\hat{p}^2+m^2)}
 \end{equation}
One can  then prove by induction that 
\begin{eqnarray}\label{covde}
\partial^{w}e^{-a^2\hat{p}^2}=\prod_{\mu=1}^4\left(\sum_{k=1}^{w_{\mu}}a_0^{w_{\mu}-k}a^k\  P_{k,\mu}\left(\cos \frac{a_0p_{\mu}}{2},
\sin \frac{a_0p_{\mu}}{2}\right)\tilde{P}_{k,\mu}\left(a\hat{p}_{\mu}\right)\right)e^{-a^2\hat{p}^2} \quad
\end{eqnarray}
Here $P, \tilde P$ are real polynomials which we do not specify.
Using (\ref{covde}) together with $a_0 \leq a$, we obtain the following bound on the propagator and its derivatives
\begin{eqnarray}\label{inecov}
\left | \partial^w \partial_{1/a}C^{a_0,a}(\hat{p})\right | \leq a^{|w|+3} \mathcal{P}_1(a|\hat{p}|)e^{-a^2(\hat{p}^2+m^2)}
\end{eqnarray}
Using (\ref{covde}) and (\ref{mbound}) below one can also show that
\begin{eqnarray}
\left | \partial^w \partial_{1/a}C^{a_0,a}(\hat{p})\right | \leq \left(\frac{1}{a}+m\right)^{-|w|-3}\mathcal{P}_2\left(\frac{a|{p}|}{1+am}\right)
\label{propder}
\end{eqnarray}
Both bounds are expressed in terms of suitable  polynomials ${\cal P}_1,\ {\cal P}_2\,$ with nonnegative coefficients.\\
The following lemma shows how to bound integrals of powers of momenta multiplied
by the exponential appearing in the regularized propagator
\begin{lemma}\label{lemma1}
$\forall \alpha \in \mathbb{N}$ ,$\ \exists C_{\alpha}>0$ independent of $a$ and $a_0$ such that:
\begin{equation}{}\label{borne}
a^4\int_{\mathcal{B}_{a_0}}e^{-a^2\hat{k}^2}\left(a|k|\right)^{\alpha}dk\leq C_{\alpha}
\end{equation}
\end{lemma}{}
\begin{proof}
It is sufficient to bound
\begin{equation}{}\label{1D}
a\int_{0}^{\frac{\pi}{a_0}}e^{-a^2\hat{k}^2}\left(ak\right)^{\alpha}dk
\end{equation}
uniformly with respect to $a$ and $a_0$. Using that $ \forall x \in \left[0,\frac{\pi}{2}\right]$ we have  $\sin x\geq \frac{2}{\pi}x\,$,
one obtains
\begin{equation}{}
 \!\!\!\!\!\! \!\!\!\!\!\! \!\!\!\!\!\! 
a\int_{0}^{\frac{\pi}{a_0}}e^{-a^2\hat{k}^2}\left(ak\right)^{\alpha}dk \ \leq\  a\int_{0}^{\frac{\pi}{a_0}}e^{-\frac{a^2 k^2}{\pi^2}}\left(ak\right)^{\alpha}dk
\  \leq\int_{0}^{\infty}e^{-\frac{ u^2}{\pi^2}}u^{\alpha}du \ \le \ C_{\alpha}
\end{equation}
\end{proof}
\noindent
When studying the restoration of rotation invariance we will also have to bound differences
of derived propagators, where one of them has undergone an arbitrary rotation $O\in O(4)$. 
The following lemma permits to bound these differences
\begin{lemma}
\label{lemma2}
For all $w \in \mathbb{N}^4$, for all $p \in \mathcal{B}_{\alpha a_0}$ for some $\alpha>0$ holds
\begin{equation}
    \fl \left|\partial^{w}\partial_{1/a}C^{a_0,a}(\hat{p})-\partial^{w}\partial_{1/a}C^{a_0,a}(\hat{p}^O)\right|\leq a_0 \left(\frac{1}{a}+m\right)^{-2-|w|}\mathcal{P}\left(\frac{a|p|}{1+am}\right)
\end{equation}
Here $\hat{p}^O:= \hat p(Op)$.
\end{lemma}
\begin{proof}
If $|\hat{p}^O|\geq |\hat{p}|$ we write
\begin{equation}\label{fac}
\fl \partial^w\left(e^{-a^2(\hat{p}^2+m^2)}-e^{-a^2((\hat{p}^O)^2+m^2)}\right)
\,=\,
\partial^w\ e^{-a^2(\hat{p}^2+m^2)}\left(1-e^{-a^2((\hat{p}^O)^2-\hat{p}^2)}\right)
\end{equation}
In case $|\hat{p}^O|\leq |\hat{p}|$, we factorize instead $e^{-a^2((\hat{p}^O)^2+m^2)}$ and follow again the subsequent reasoning.
By the Leibniz formula, we obtain
\begin{equation*}
\fl \partial^w\ e^{-a^2(\hat{p}^2+m^2)}\left(1-e^{-a^2((\hat{p}^O)^2-\hat{p}^2)}\right)\,=\sum_{w_1+w_2=w}c_{w_i}\partial^{w_1}e^{-a^2(\hat{p}^2+m^2)}\partial^{w_2}\left(1-e^{-a^2((\hat{p}^O)^2-\hat{p}^2)}\right)
\end{equation*}
The first factor in each entry in the sum can be bounded as in (\ref{propder}).
As regards  the second factor we first consider the exponential without derivatives
\[
1-e^{-a^2[(\hat{p}^O)^2-\hat{p}^2]}
 \]
We can rewrite  the exponent as
\begin{equation}
\label{dar}
a^2\,[(\hat{p}^O)^2-\hat{p}^2]=\frac{2a^2}{a_0^2}\sum_{\mu=1}^4\left[\cos(a_0(Op)_\mu)-\cos(a_0 p_{\mu})\right]
\end{equation}
\begin{equation}
\!\!\!\!\!\!\!\!\!\!\!\!\!\!\!\!\!\!\!\!\!\!\!\!\!\!\!
\,=\,
2\frac{a_0}{a} \sum_{\mu=1}^4 \int_0^1 dt \frac{(1-t)^2}{2!}  \left[\, [a (Op)_\mu]^3
\cos^{(3)} [t\, a_0  \, (Op)_\mu]\,-\,
[a p_\mu]^3  \cos^{(3)} [t\, a_0  \, p_\mu] \right]
\label{t3} 
 \end{equation}
 We used a Taylor formula with integrated remainder around $0$ for both cosine functions
 and the fact that the constant and quadratic terms in the difference of the two cosine functions cancel. 
 The statement of the lemma is then a consequence
of the following facts\\
a) 
\begin{equation}\label{prp}
\left| \,\partial^w  a^2\,[(\hat{p}^O)^2-\hat{p}^2]\, \right| \le \frac{a_0}{a} \ a^{-|w|}\ {\cal P}(a |p|)
\end{equation}
This follows directly from (\ref{dar}), (\ref{t3}). The degree of the polynomial $\mathcal{P}$ can be chosen to be less equal than 3.
\\
b) 
\[
\left|\, e^{-f(x)} -1\, \right| \le f(x) \  \mbox{ for } \  f(x) \ge 0
\]
c) 
\[
\!\!\!\!\!\!\!\!\!\!
\partial^w e^{-a^2[(\hat{p}^O)^2-\hat{p}^2]} \,=\, 
\]
\[
\!\!\!\!\!\!\!\!\!\!
\frac{a_0^2}{a^2} \ a^{|w|} \ 
P(\{a\,p_\mu ,\int_0^1 dt (1-t)^2  \cos^{(3)} (t\, a_0 p_\mu), \int_0^1 dt (1-t)^2  \cos^{(3)} (t\, a_0 (Op)_\mu),\frac{a_0}{a} \})
\]
This statement follows by induction on $|w|$ from (\ref{dar}), (\ref{t3}). The polynomial
$P$ (whose coefficients are real but may have either sign) is at most of degree
$3 |w|$. The coefficients do not depend on $a_0,\,a,\,p$. \\
d)
The inequality
\begin{equation}
\label{mbound}
e^{-a^2 m^2} \,\le \, \frac{C(n)}{(1+am)^n}
\end{equation}
which holds for any $n \in \mathbb{N}$ and suitable positive $C(n)$ can be used to turn powers
of $a$ or of $a|p|$ into powers of $a/(1+am)$ or $a|p|/(1+am)\,$.
\end{proof}


\subsection{Renormalizability}

A simple inductive proof of  the renormalizability of $\phi_4^4$ theory, regularized by a UV-cutoff has been exposed several times 
in the literature \cite{6,5}. Our proof follows the same line of reasoning.  New difficulties  arise due to the particular 
form of the lattice propagator (\ref{propagator}) that breaks Euclidean symmetry. The boundary conditions 
following from (\ref{bareinter}) are
\begin{eqnarray}{}
{\partial}^{w}\mathcal{L}_{l,n}^{a_0,a_0}\left({p}_1,\cdots,{p}_n\right)=0,\qquad n+|w|>4\quad \mathrm{such~that}~n \neq 2\label{bc0}\\
\partial^w \mathcal{L}_{l,2}^{a_0,a_0}(p,-p)=b_l(a_0)\partial^w\hat{p}^2,\qquad \forall |w|\geq 3\label{bc1}
\end{eqnarray}
As compared to continuum theory \cite{6}, note that the boundary conditions (\ref{bc1}) are not equal to zero. 
For terms with $n+|w|\leq 4$, the boundary conditions are explicitly fixed by ($a_0$-independent) renormalization conditions imposed for the fully integrated theory at $a=\infty\,$: 
\begin{equation}\label{bc2}
  \mathcal{L}_4^{a_0,\infty}(0,\cdots,0)=f,\qquad \mathcal{L}_2^{a_0,\infty}(0,0)=0, \qquad {\partial}_{{p}^2}\mathcal{L}_2^{a_0,\infty}(0,0)=0   
\end{equation}
The renormalization point is chosen at zero momentum for simplicity (BPHZ renormalization conditions).\\
The induction hypotheses to be proven are
\begin{theorem}
\label{propbo}
For all $l\in \mathbb{N}^*$, $n \in \mathbb{N}$, $w$ and for $0\leq a_0 \leq a$, $a_0 < \frac{1}{m}$  holds\\ 
A) Boundedness in the UV-cutoff
\begin{eqnarray}\label{bou}
\fl \left | {\partial}^{w}\mathcal{L}_{l,n}^{a_0,a}({p}_{1},\cdots,{p}_{n}) \right | \leq \left( \frac{1}{a}+m \right)^{4-n-|w|}
 \mathcal{P}_1\left( \log \frac{1+am}{am} \right) \mathcal{P}_2 \left( \left \{ \frac{a|{p}_i|}{1+am} \right \} \right)
\end{eqnarray}
B) Convergence in the UV-limit
\begin{eqnarray}\label{con0}
\fl \left | \partial_{{1}/{a_0}}{\partial}^{w}\mathcal{L}_{l,n}^{a_0,a}({p}_{1},\cdots,{p}_{n}) \right | \leq \frac{\left( \frac{1}{a}+m \right)^{5-n-|w|}}
{\left( \frac{1}{a_0}+m \right)^2} \mathcal{P}_3\left( \log \frac{1+a_0m}{a_0m} \right) \mathcal{P}_4 \left( \left \{ \frac{a|{p}_i|}{1+am} \right \} \right)
\end{eqnarray}
where $(p_i)_{1\leq i \leq n}\in \mathcal{B}_{ a_0}$ and $p_1+\cdots+p_n\equiv 0 \left[\frac{2\pi}{a_0}\right]$.
Here and in the following the $\mathcal{P}, \ \mathcal{P}_i$ denote (each time they appear possibly new) polynomials with nonnegative coefficients. 
The coefficients depend on $l,n,|w|$, but not on $ m, \left \{ p_i \right \}$, $a$, $a_0$. For $l=0$, all polynomials $\mathcal{P}_1$,  $\mathcal{P}_3$
reduce to $1$.
\end{theorem}
\noindent
{\it Remarks:} 
We will prove Theorem \ref{propbo} for $p_i \in \mathcal{B}_{a_0}$ but it is possible to extend it to $(p_i)_{1\leq i \leq n}\in \mathbb{R}^4$. 
Since
$\mathcal{L}^{a_0,a}_{l,n}$ is $\frac{2\pi}{a_0}$-periodic, $\mathcal{L}_{l,n}^{a_0,a}(p_1,\cdots,p_n)$ such that 
$p_i \in \mathcal{B}_{k_i a_0}:=\left]-\frac{(2k_i+1)\pi}{a_0},\frac{(2k_i+1)\pi}{a_0}\right[,\ k_i\in \mathbb{Z}^4$ 
and $\sum_{i=1}^n p_i\equiv 0 \left[\frac{2\pi}{a_0}\right]$,  also verifies the flow equations (\ref{feq}) 
with the same boundary conditions, as we will see later, and therefore it verifies Theorem 1. The extension to the 
boundaries of the extended Brillouin zones $\mathcal{B}_{k_i,a_0}$  is performed  using 
the continuity of $\mathcal{L}^{a_0,a}_{l,n}$ w.r.t. $p_i$ and taking the limits $p_i\rightarrow \frac{k_i \pi}{a_0}$ in (\ref{bou}).
The fact that 
 $\mathcal{L}^{a_0,a}_{l,n}$ is $\mathcal{C}^{\infty}$ w.r.t. $p_i$ and that it is $2\pi/a_0$-periodic can be proven 
 inductively using the flow equations and that the propagator and the boundary conditions are $2\pi/a_0$-periodic and $\mathcal{C}^{\infty}$. We will not prove it here.\\
 It is also possible to prove a stronger version of Theorem \ref{propbo}, replacing
 $\mathcal{P}\left( \left \{ \frac{a|{p}_i|}{1+am} \right \} \right)$ by 
 $\mathcal{P} \left( \left \{ \frac{a|\hat{p}_i|}{1+am} \right \} \right)$. \\

\noindent The statement (\ref{con0}) implies that for sufficiently small  $a_0$ and suitable $\nu >0$
\begin{eqnarray}\label{con}
\fl \left | \partial_{{1}/{a_0}}{\partial}^{w}\mathcal{L}_{l,n}^{a_0,a}({p}_{1},\cdots,{p}_{n}) \right |
 \leq a_0^2\left( \frac{1}{a}+m \right)^{5-n-|w|} \left( \log \frac{1+a_0m}{a_0m} \right)^{\nu} \mathcal{P}_4 \left( \left \{ \frac{a|{p}_i|}{1+am} \right \} \right)
\end{eqnarray}
Integration of the bound (\ref{con}) over the lattice cutoff ${1}/{a_0}$ immediately proves the convergence of all 
$\mathcal{L}_{l,n}^{a_0,a}({p}_1,\cdots,{p}_n)$ for fixed $a$ to finite limits when $a_0\rightarrow 0$. In particular, one obtains 
for all $\hat{a}_0<a_0$ and $(p_i)_{1\leq i \leq n} \in \mathcal{B}_{a_0}$,
\begin{equation}\label{cauchy}
\fl \left | \mathcal{L}_{l,n}^{a_0,\infty} \left( {p}_1,...,{p}_n \right) - \mathcal{L}_{l,n}^{\hat{a}_0,\infty}\left({p}_1,...,{p}_n\right) \right | < a_0 m^{5-n}\left(\log\frac{1}{a_0m}\right)^{\nu}\mathcal{P}_5\left( \left \{\frac{|{p}_i|}{m}\right \} \right)
\end{equation}
Thus, due to the Cauchy criterion in $\mathcal{C}^{\infty}(\mathbb{R}^+)$ (w.r.t. to $a_0$) finite limits exist to all loop orders $l$.
\begin{proof}
The statement (\ref{bou}) has to be obtained first. The induction scheme to prove the statements proceeds upwards in $l$, 
for given $l$ upwards in $n$, and for given $(n,l)$ downwards in $|w|$ starting from some arbitrary $|w_{\max}|\geq 3$. 
The induction works because the terms on the r.h.s. of the FE always are prior to the one of the l.h.s. in the inductive order. 
So the bounds (\ref{bou}) and (\ref{con0}) may be used as an induction hypothesis on the r.h.s. Then we integrate the FE, 
where the terms with $n+|w|\geq 5$ are integrated down from $1/a_0$ to $1/a$ because of the boundary conditions (\ref{bc0})-(\ref{bc1}) and the  terms with $n+|w|\leq 4$ at the renormalization point are integrated upwards from $0$ to $1/a$ since we have (\ref{bc2}). Therefore, we can write
\begin{eqnarray}\label{renopoint}
\fl {\partial}^{w}\mathcal{L}_{l,n}^{a_0,a}\left(0,\cdots,0\right)={\partial}^{w}\mathcal{L}_{l,n}^{a_0,\infty}\left(0,\cdots,0\right)+\int_{0}^{1/a} 
d\lambda\ \partial_{\lambda}\ {\partial}^{w}\mathcal{L}_{l,n}^{a_0,\frac{1}{\lambda}}\left(0,\cdots,0\right)
\end{eqnarray}
Once a bound has been obtained at the renormalization point, it is possible to move away from the renormalization point using the integrated Taylor formula,
\begin{eqnarray}
\label{taylorfor}
\fl{{\partial}^{w}\mathcal{L}_{l,n}^{a_0,a}\left({p}_1,\cdots,{p}_n\right)={\partial}^{w}\mathcal{L}_{l,n}^{a_0,a}\left(0,\cdots,0\right)+\sum_{i=1}^n \sum_{\mu=1}^4 {p}_{i,\mu} \int_0^1 dt \left( {\partial}_{{p}_{i,\mu}} {\partial}^{w}\mathcal{L}_{l,n}^{a_0,a}\right)\left(tp_1,\cdots,tp_n\right)}
\end{eqnarray}
\begin{itemize}
    \item[(A)] \underline{Boundedness}: 
    To start the induction, we prove the bound (\ref{bou}) at the tree level. The classical interaction contains no terms linear or quadratic in the fields. To bring the system of flow equations to bear, however, at first the crucial properties,
$$\mathcal{L}^{a_0,a}_{0,2}({p},-{p})=0,\qquad \mathcal{L}^{a_0,a}_{0,4}({p}_1,\cdots,{p}_4)=f$$
have to be inferred directly from the representation (\ref{fl}). 
Since the $\mathbb{Z}_2$-symmetry  $\phi\rightarrow -\phi$, is not broken by the renormalization procedure, we note 
$$\mathcal{L}^{a_0,a}_{l,n}({p}_1,\cdots,{p}_n)=0,\qquad \forall n \ \mbox{odd}\, ,\  \forall l$$
Thus, the bound evidently holds for $n+|w|\leq 4$. For  $n+|w|>4$ (the irrelevant cases)
proceed inductively ascending in $n$. For given $n$ the various $w$ dealt with in arbitrary order, by integrating the 
respective flow equation (\ref{feq}) from the initial point $1/a_0$.\\
Using the induction hypothesis for $\mathcal{L}_{0,n_1+1}^{a_0,a}$ and $\mathcal{L}_{0,n_2+1}^{a_0,a}$, 
and (\ref{propaa}), (\ref{propder}) we obtain a bound for 
the quadratic part of the r.h.s. of (\ref{feq}) 
\begin{eqnarray}
\fl\left | \partial^{w_1}\mathcal{L}_{0,n_1+1}^{a_0,a}(p_1,\cdots,p_{n_1},p)\partial^{w_3} 
\partial_{1/a}C^{a_0,a}(\hat{p})\partial^{w_2}\mathcal{L}_{0,n_2+1}^{a_0,a}(-p,p_{n_1+1},\cdots,p_n)\right|\nonumber\\
\leq \left(\frac{1}{a}+m\right)^{4-n-|w|-1}\mathcal{P}\left(\left \{ \frac{|{p}_i|}{\lambda+m}\right \}\right)
\end{eqnarray}
Therefore
\begin{eqnarray}\label{Inte}
\left | \partial_{\lambda}{\partial}^{w}\mathcal{L}_{0,n}^{a_0,\frac{1}{\lambda}}\left({p}_1,\cdots,{p}_n\right) \right |
     \leq \left( \lambda+m\right)^{4-n-|w|-1}\mathcal{P}\left( \left \{ \frac{|{p}_i|}{\lambda+m}\right \}\right) 
\end{eqnarray}
This proves (\ref{bou}) at the tree order.\\
To generate inductively the bounds (\ref{bou}) for higher loop orders, we use them in bounding the r.h.s of the FE (\ref{feq}), 
together with the bound (\ref{inecov}) in the linear and in the quadratic term respectively. For the linear term of the r.h.s. of FE,
 we use the induction hypothesis for $\partial^{w}\mathcal{L}_{l-1,n+2}^{a_0,a}$,
 and we obtain the upper bound 
$$
\displaystyle \int_{k,\mathcal{B}_{a_0}} (2a^3)e^{-a^2(\hat{k}^2+m^2)}\mathcal{P}\left(\frac{a|k|}{1+am},\left \{ \frac{|{p}_i|}{\lambda+m}\right \}\right) $$
Using lemma \ref{lemma1} this can be turned into the bound
$$
\displaystyle \int_{k,\mathcal{B}_{a_0}} (2a^3)e^{-a^2(\hat{k}^2+m^2)}
\mathcal{P}\left(\frac{a|k|}{1+am},\left \{ \frac{|{p}_i|}{\lambda+m}\right \}\right) \leq \frac{1}{a}\tilde{\mathcal{P}}\left(\left \{ \frac{|{p}_i|}{\lambda+m}\right \}\right)$$
Hence
\begin{eqnarray*}
\fl \displaystyle \int_{k,\mathcal{B}_{a_0}} \partial_{1/a}C^{a_0,a}(\hat{k}) \left | \partial^w \mathcal{L}^{a_0,a}_{l-1,n+2}(-k,\cdots,k)\right |\\
\leq \left(\frac{1}{a}+m\right)^{4-n-|w|-1}\mathcal{P}_1\left(\log \frac{am+1}{am}\right)\mathcal{P}_2\left(\left \{ \frac{|{p}_i|}{\lambda+m}\right \}\right)
\end{eqnarray*}
For the quadratic part of the flow equations (\ref{feq}), we use the induction hypothesis for
 $\partial^{w_1} \mathcal{L}_{l_1,n_1+1}^{a_0,a}$ and $\partial^{w_2} \mathcal{L}_{l_2,n_2+1}^{a_0,a}$ 
 together with the bound (\ref{covde}) and we obtain 
\begin{eqnarray}
\fl\left | \partial^{w_1}\mathcal{L}_{l_1,n_1+1}^{a_0,a}(p_1,\cdots,p_{n_1},p)\partial^{w_3} 
\partial_{1/a}C^{a_0,a}(\hat{p})\partial^{w_2}\mathcal{L}_{l_2,n_2+1}^{a_0,a}(-p,p_{n_1+1},\cdots,p_n)\right|\nonumber\\
\leq \left(\frac{1}{a}+m\right)^{4-n-|w|-1}\mathcal{P}_1\left(\log \frac{1+am}{am}\right)\mathcal{P}_2\left(\left \{ \frac{|{p}_i|}{\lambda+m}\right \}\right)
\end{eqnarray}
Therefore
\begin{eqnarray}\label{Integ}
\fl \left | \partial_{1/a}\partial^{w}\mathcal{L}_{l,n}^{a_0,a}(p_1,\cdots,p_{n}) \right | \\\nonumber
\leq \left(\frac{1}{a}+m\right)^{4-n-|w|-1}\mathcal{P}_1\left(\log \frac{am+1}{am}\right)\mathcal{P}_2\left(\left \{ \frac{|{p}_i|}{\lambda+m}\right \}\right)
\end{eqnarray}
Following the order of the induction stated before, for the irrelevant cases $n+|w|\geq 5$ 
the bound (\ref{Integ}) is integrated downwards from $1/a$ to $1/a_0$. 
For $n+|w|\geq 5$ such that $n\neq 2$, integrating from $1/a$ to $1/a_0$ yields 
\begin{eqnarray*}
 \fl \left| {\partial}^{w}\mathcal{L}_{l,n}^{a_0,a}\left({p}_1,\cdots,{p}_n\right) \right|
   \leq \displaystyle \int_{1/a}^{1/a_0} d\lambda \left(\lambda+m\right)^{4-n-|w|-1}
   \mathcal{P}_1\left(\log \frac{\lambda+m}{m}\right)\mathcal{P}_2\left(\left \{ \frac{|{p}_i|}{\lambda+m}\right \}\right)
\end{eqnarray*}
We now have, see \cite{6}
$$\int_{1/a}^{1/a_0} d\lambda \left( \lambda+m\right)^{4-n-|w|-1} \mathcal{P} \left( \log \frac{\lambda+m}{m}\right)
< \left( \frac{1}{a}+m\right)^{4-n-|w|}\tilde{\mathcal{P}} \left( \log \frac{1+am}{am}\right)$$
For the particular case $(n,|w|)=(2,2)$,  (\ref{Integ}) is integrated from $0$ to $1/a_0$ at zero momenta,
\begin{equation*}
\fl \eqalign{\left | \partial_{p^2}\mathcal{L}_{l,2}^{a_0,a_0}(0,0)-\partial_{p^2}\mathcal{L}_{l,2}^{a_0,\infty}(0,0)\right |
   \leq \displaystyle \int_{0}^{\frac{1}{a_0}}d\lambda\left(\lambda+m\right)^{-1}\mathcal{P}\left(\log\frac{\lambda+m}{m}\right)
\leq \mathcal{P}\left(\log\frac{1+a_0m}{a_0m}\right)}
\end{equation*}
This gives 
\begin{equation}
\left | b_l(a_0) \right | \leq \mathcal{P}\left(\log\frac{1+a_0m}{a_0m}\right)
\label{bla0}
\end{equation}
It then follows from (\ref{bc1}) that the $2$-point function and its derivatives at $a=a_0$ can be bounded
$$ \left |\partial^w\mathcal{L}_{l,2}^{a_0,a_0}(p,-p)\right|\leq 2|b_l(a_0)|a_0^{|w|-2} \ C $$
 for some positive constant $C$ depending on $|w|$, which implies for all $|w| \geq 3$  
\[
\!\!\!\!  \!\!\!\!\!\! \!\!\!\!\!\! \!\!\!\!\!\!  \!\!\!\!\!\! \!\!\!\!\!\! \!\!\!\!\!\! 
\left | \partial^w\mathcal{L}_{l,2}^{a_0,a_0}(p,-p)\right| \leq \left(\frac{1}{a_0}+m\right)^{2-|w|}\mathcal{P}\left(\log\frac{1+a_0m}{a_0m}\right)
\leq \left(\frac{1}{a}+m\right)^{2-|w|}\mathcal{P}\left(\log\frac{1+am}{am}\right)
\]
Integrating the inductive bound  from $1/a$ to $1/a_0$ for $n=2,~|w|\geq 3$ then gives
\begin{equation*}
\fl \eqalign{   \left| {\partial}^{w}\mathcal{L}_{l,2}^{a_0,a}\left({p},-{p}\right) \right| & \leq \displaystyle \int_{1/a}^{1/a_0}d\lambda 
\left | \partial_{\lambda}{\partial}^{w}\mathcal{L}_{l,2}^{a_0,\frac{1}{\lambda}}\left({p},-{p}\right) \right |
+\left| {\partial}^{w}\mathcal{L}_{l,2}^{a_0,a_0}\left({p},-{p}\right) \right|\cr
  & \leq  \left(\frac{1}{a}+m\right)^{2-|w|}\mathcal{P}_1\left(\log \frac{1+am}{am}\right)\mathcal{P}_2\left(\left \{ \frac{|{p}_i|}{\lambda+m}\right \}\right)}
\end{equation*}
For the relevant terms ($n+|w|\leq 4$), we start with the case ($n=2,|w|=2$) and continue to ($n=2,|w|=1$) and ($n=2,w=0$). 
Bounding equation (\ref{renopoint}) in absolute value, we obtain using the bound (\ref{Integ}) at vanishing momenta:
\begin{eqnarray}
   \fl \eqalign{\left | \int_{0}^{1/a}d\lambda \partial_{\frac{1}{\lambda}}
   {\partial}^{w}\mathcal{L}_{l,n}^{a_0,\frac{1}{\lambda}}\left(0,\cdots,0\right) \right | 
   & \leq \int_{0}^{1/a} d\lambda \left( \lambda+m\right)^{4-n-|w|-1}\mathcal{P}\left( \log \frac{1+am}{am}\right)\cr
    & \leq \left( \lambda+m\right)^{4-n-|w|} \mathcal{P} \left( \log \frac{1+am}{am}\right)} \label{renobound} 
\end{eqnarray}
Hence, the assertion (\ref{bou}) is established at the renormalization point. In each case extension to general momenta 
via (\ref{taylorfor}) is guaranteed by the bounds established before. This
concludes the proof of (\ref{bou}).
\item[(B)] \underline{Convergence}: The bound (\ref{con}) follows on applying the same inductive scheme 
to bound the solutions of the FE, integrated over $1/a$ and then derived w.r.t. $1/a_0$. The proof is analogous to \cite{3}, \cite{6} 
apart from the changes  induced by  the lattice momenta $\hat{p}$ which were dealt with in the proof of (\ref{bou}).
\end{itemize}
\end{proof}


\section{Restoration of  $O(4)$ symmetry}\label{Sec4}


\subsection{The flow equations}

The lattice breaks the rotation and translation symmetries. In order to define the rotated scalar field on the lattice, we consider the rotated lattice
$$\Lambda_{a_0}^O:=O\Lambda_{a_0},~~~~O\in \mathrm{O}(4)$$
The rotated scalar field $\hat{\phi}^{O}_{a_0}$ is defined by $$\hat{\phi}_{a_0}^O:=\hat{\phi}|_{\Lambda_{a_0}^O}$$
where $\hat{\phi}$ is the continuum scalar field. For our purposes, $\hat{\phi}^O_{a_0}$ is considered to live in 
$l_2(\Lambda_{a_0}^O)$ and the Brillouin zone associated to the rotated lattice $\Lambda_{a_0}^O$ is 
$$\mathcal{B}_{a_0}^O:=O\left(\left]-\frac{\pi}{a_0},\frac{\pi}{a_0}\right[^4\right)$$
The Fourier transform of $\hat{\phi}_{a_0}^O$ is defined by 
$$\phi_{a_0}^O(p):=a_0^4\sum_{x \in \Lambda_{a_0}^O}e^{-ip\cdot x}\hat{\phi}_{a_0}^O(x)$$
The inverse Fourier transform is defined by 
$$\hat{\phi}_{a_0}^O(x):=\int_{\mathcal{B}_{a_0}^O}\frac{d^4p}{(2\pi)^4}\hat{\phi}^O_{a_0}(p)e^{ip\cdot x}$$
such that the Plancherel identity is preserved.\\
The bare action associated to the rotated field is defined by 
\begin{eqnarray*}
\fl L_O^{a_0,a_0}(\hat{\phi}^O_{a_0}):=a_0^4 \sum_{x \in \Lambda_{a_0}^O}\left \{ \frac{f}{4!}\left(\hat{\phi}^O_{a_0}\right)^4+d(a_0)\left(\hat{\phi}^O_{a_0}\right)^2+b(a_0)\left(\hat{\partial}^O_{\mu,a_0}\hat{\phi}^O_{a_0}\right)^2+c(a_0)\left(\hat{\phi}^O_{a_0}\right)^4\right \}
\end{eqnarray*}
Note that the counter terms $d(a_0)$, $b(a_0)$ and $c(a_0)$ are the same as in the bare action $L^{a_0,a}$, 
since they are space-time independent and depend only on the spacing $a_0$ between the points of the lattice.
The lattice derivative on the rotated lattice is defined as follows for $ \hat{\phi}^O_{a_0}\in l_2\left(\Lambda_{a_0}^O\right)$ 
$$\left(\hat{\partial}^O_{\mu,a_0}\hat{\phi}^O_{a_0}\right)(x):=\frac{\hat{\phi}^O_{a_0}(x+a_0e_{\mu}^O)-\hat{\phi}^O_{a_0}(x)}{a_0},
~~~~x\in \Lambda_{a_0}^O$$
where $e_{\mu}^O:=Oe_{\mu}$ is the rotated unit vector in the $\mu^{\mathrm{th}}$ direction. \\
The flowing propagator is defined by 
$$\hat{C}^{O,a_0,a}(x,y):=\int_{\mathcal{B}_{a_0}^O}\frac{d^4p}{(2\pi)^4}e^{ip\cdot(x-y)}C^{a_0,a}(p)$$
where $C^{a_0,a}$ is defined as before,
$$C^{a_0,a}(p):=\frac{1}{\hat{p}^2+m^2}\left(e^{-a_0^2(\hat{p}^2+m^2)}-e^{-a^2(\hat{p}^2+m^2)}\right)$$
The lattice momentum $\hat{p}$ was defined in (\ref{lattMom}).
The derivation of the FE corresponding to the rotated field follows the same steps as before, starting from the functional integral
\begin{equation}\label{rotfun}
e^{-\frac{1}{\hbar}\left(L_{O}^{a_0,a}(\hat{\phi}_{a_0}^O)+I^{a_0,a}\right)}:=\int d\mu^O_{a_0,a}\left(\Phi\right)e^{-\frac{1}{\hbar}L_{O}^{a_0,a_0}(\hat{\phi}_{a_0}^O+\Phi)}
\end{equation}
where $d\mu_{a_0,a}^O$ is uniquely defined by the covariance operator $\hat{C}^{O,a_0,a}$,
$$\int d\mu_{a_0,a}^O(\Phi) e^{\langle \Phi,J \rangle_{l_2\left(\Lambda_{a_0}^O\right)}}:=e^{\frac{1}{2}\langle J,\hat{C}^{O,a_0,a}J\rangle_{l_2\left(\Lambda_{a_0}^O\right)}},~~~J\in {l_2\left(\Lambda_{a_0}^O\right)}$$
In terms of momenta in $\mathcal{B}_{a_0}$, the propagator $C^{a_0,a}$ has the following form
$$C^{a_0,a}(Op)=\frac{1}{(\hat{p}^O)^2+m^2} \left(e^{-a_0^2((\hat{p}^O)^2+m^2)}-e^{-a^2((\hat{p}^O)^2+m^2)}\right) $$
The FE are obtained by differentiating (\ref{rotfun}) w.r.t. $1/a$,
\begin{equation}\label{FEO}
    \fl \partial_{1/a}L^{a_0,a}_O=\frac{\hbar}{2}\langle \frac{\delta}{\delta \hat{\phi}_{a_0}^O},\dot{C}^{a_0,a}*\frac{\delta}{\delta \hat{\phi}_{a_0}^O}\rangle_{l_2\left(\Lambda_{a_0}^O\right)} L^{a_0,a}_O -\frac{1}{2} \langle \frac{\delta L^{a_0,a}_O}{\delta \hat{\phi}_{a_0}^O}, \dot{C}^{a_0,a}* \frac{\delta L^{a_0,a}_O}{\delta \hat{\phi}_{a_0}^O} \rangle_{l_2\left(\Lambda_{a_0}^O\right)}
\end{equation}
We expand in a formal power series w.r.t. $\hbar$ to select the loop order,
$$L^{a_0,a}_O(\hat{\phi}_{a_0}^O)=\sum_{l=0}^{+\infty}\hbar^l L_{O,l}^{a_0,a}(\hat{\phi}_{a_0}^O)$$
From $L_{O,l}^{a_0,a}$  we obtain the CAS of loop order $l$ in momentum space $\mathcal{B}_{a_0}$ as
$$\delta^{4}_{\left[\frac{2\pi}{a_0}\right]}(Op_1+\cdots+Op_n)\mathcal{L}^{a_0,a,O}_{l,n}(Op_1,\cdots,Op_n):=(2\pi)^{4(n-1)}\delta_{\phi^O_{a_0}(Op_1)}
\cdots \delta_{\phi^O_{a_0}(Op_n)}L^{a_0,a}_{O,l}|_{\phi^O_{a_0}\equiv 0}$$
From the functional flow equations (\ref{feq}), we obtain the perturbative flow equations for the CAS $n$-point functions
\begin{eqnarray}\label{FErot}
   \fl  \partial_{1/a}\partial^w\mathcal{L}_{l,n}^{a_0,a,O}(Op_1,\cdots,Op_n)
    \\=\frac{1}{2}\int_{k,\mathcal{B}_{a_0}}\partial^w\mathcal{L}_{l-1,n+2}^{a_0,a,O}(Ok,Op_1,\cdots,Op_n,-Ok) \partial_{1/a}C^{a_0,a,O}(\hat{k}^O)\nonumber
    \\-\frac{1}{2}\sum_{l_1,l_2}^{'}\sum_{n_1,n_2}^{'}\sum_{ w_1,w_2,w_3}^{'}c_{w_i}\left[\partial^{w_1}\mathcal{L}_{l_1,n_1+1}^{a_0,a,O}(Op_1,\cdots,Op_{n_1},Op)\partial^{w_3}\partial_{1/a}C^{a_0,a,O}(\hat{p}^O)\right. \nonumber\\
    \left. \partial^{w_2}\mathcal{L}_{l_2,n_2+1}^{a_0,a,O}(-Op,\cdots,Op_{n})\right]_{rsy}\nonumber
\end{eqnarray}{}
$$Op\equiv-Op_1-\cdots-Op_{n_1}\equiv Op_{n_1+1}+\cdots+Op_n~\left[\frac{2\pi}{a_0}\right]\ ,~~(p_i)_{1\leq i \leq n} \in \mathcal{B}_{a_0}$$
where we used the same conventions as in (\ref{feq}).


\subsection{Proof of rotation symmetry restoration}

The $O(4)$-symmetry is restored for $a_0\rightarrow 0$ if and only if
$\forall (p_i)_{1\leq i \leq n}\in \mathbb{R}^4\,,\  \forall O \in O(4) \ \, \exists \tilde{a}_0\geq 0$,
\begin{equation}\label{limit}
    \lim_{a_0\rightarrow 0,0\leq a_0\leq \tilde{a}_0}\lim_{a \rightarrow \infty}
     \left(\mathcal{L}^{a_0,a}_{l,n}(p_1,...,p_n)-\mathcal{L}^{a_0,a,O}_{l,n}(Op_1,...,Op_n)\right)=0
\end{equation}{}
Here we introduced the parameter $\tilde{a}_0$ as in (\ref{limit1}). For $(p_i)_{1\leq i \leq n}\in \mathcal{B}_{a_0}$ we thus define
\begin{eqnarray*}
\partial^{w}\mathcal{D}_{l,n}^{a_0,a}({p}_1,...,{p}_n):={\partial}^{w}\mathcal{L}_{l,n}^{a_0,a}(p_1,\cdots,p_n)-{\partial}^{w}\mathcal{L}_{l,n}^{a_0,a,O}(Op_1,\cdots,Op_n)
\end{eqnarray*}
From the flow equations (\ref{FErot}) and (\ref{feq}), we can derive a FE for $\partial^{w}\mathcal{D}_{l,n}^{a_0,a}(p_1,\cdots,p_n)$ :
{\small
\begin{eqnarray}\label{Dln}
\fl \partial_{1/a}\partial^{w}\mathcal{D}_{l,n}^{a_0,a}({p}_1,...,{p}_n)
\,=\,
\frac{1}{2}\int_{k,\mathcal{B}_{a_0}} \partial_{1/a}C^{a_0,a}(\hat{k}) \partial^{w}\mathcal{D}_{l-1,n+2}^{a_0,a}(k,{p}_1,...,{p}_n,-k)\\
\fl +
\frac{1}{2}\int_{k,\mathcal{B}_{a_0}}  {\partial}^{w}\mathcal{L}_{l-1,n+2}^{a_0,a}(Ok,Op_1,\cdots,Op_n,-Ok)
\left [ \partial_{1/a}C^{a_0,a}(\hat{k})-\partial_{1/a}C^{a_0,a}(\hat{k}^O)\right]\nonumber\\
\fl -\frac{1}{2}\sum^{'}_{\underset{n_1,n_2}{l_1,l_2}}\sum_{{w_1,w_2,w_3}}^{'}c_{w_i} 
\Bigl[ \partial^{w_1}\mathcal{L}_{l_1,n_1+1}^{a_0,a}(p_1,\cdots,p_{n_1})
  {\partial}^{w_3}\partial_{1/a}C^{a_0,a}(\hat{p})\partial^{w_2}\mathcal{D}^{a_0,a}_{l_2,n_2+1}(-p,\cdots,p_n)
  \nonumber\\
\fl \ +\
{\partial}^{w_1}\mathcal{D}^{a_0,a}_{l_1,n_1+1}(p_1,\cdots ,p_{n_1})
 {\partial}^{w_3}\partial_{1/a}C^{a_0,a}(\hat{p}^O){\partial}^{w_2}\mathcal{L}_{l_2,n_2+1}^{a_0,a,O}(-Op,\cdots,Op_n) 
\nonumber\\
\fl \ +\ {\partial}^{w_1}\mathcal{L}_{l_1,n_1+1}^{a_0,a}( {p}_1,\cdots, {p}_{n_1})
{\partial}^{w_3}(\partial_{1/a}C^{a_0,a}(\hat{p})-\partial_{1/a}C^{a_0,a}(\hat{p}^O))
{\partial}^{w_2}\mathcal{L}_{l_2,n_2+1}^{a_0,a,O}(-Op,\cdots,Op_n)\Bigr]_{\footnotesize{rsy}}\nonumber\\
p_{1}+\cdots +p_{n}\equiv0 \left[ \frac{2\pi}{a_0} \right]\\
Op_{1}+\cdots +Op_{n}\equiv 0\left[\frac{2\pi}{a_0}\right],~~(p_i)_{1\leq i \leq n}\in \mathcal{B}_{a_0} \nonumber
\end{eqnarray}}

\noindent
Restoration of $O(4)$-symmetry, i.e. 
$$\lim_{a_0\rightarrow0, a\rightarrow \infty}\mathcal{D}_{l,n}^{a_0,a}({p}_1,...,{p}_n)=0$$
follows from the following Theorem.
\begin{theorem}\label{O(4)-sym}
$\forall n$, $\forall w$, $\forall (p_i)_{1\leq i \leq n} \in \mathcal{B}_{a_0}$ such that $\sum_{i=1}^n p_i, \sum_{i=1}^n Op_i\equiv 0~~\left[\frac{2\pi}{a_0}\right]$,
\begin{equation}\label{Dlnine1}
\fl ~~~\left | \partial^{w}\mathcal{D}_{l,n}^{a_0,a}({p}_1,\cdots,{p}_n)\right |\leq a_0 \left( \frac{1}{a}+m \right)^{5-n-|w|} \mathcal{P}_1\left( \log \frac{1+a_0m}{a_0m}\right) \mathcal{P}_2\left( \left\{\frac{a |p_i|}{1+am}\right \}\right)
\end{equation}
where $\mathcal{P}_i$ denote polynomials with nonnegative coefficients, that depend, as well as the degree of the polynomials on $l$, $n$, $w$ but not on $m$, $\left \{ p_i \right \}$, $a$, $a_0$. 
\end{theorem}


\subsection{Proof of Theorem \ref{O(4)-sym}}

\begin{proof}
We prove (\ref{Dlnine1}) using the inductive scheme indicated previously. The only terms in which (\ref{Dlnine1}) cannot be used as an induction hypothesis are
\begin{equation}\label{b1}
    \indent \fl \int_{k,\mathcal{B}_{a_0}}\partial^w \mathcal{L}_{l-1,n+1}^{a_0,a,O}(Ok,Op_1,\cdots,p_n,-Ok)
    \left[\partial_{1/a}C^{a_0,a}(\hat{k})-\partial_{1/a}C^{a_0,a}(\hat{k}^O)\right]
\end{equation}
and
\begin{eqnarray}\label{b2}
    \fl \partial^{w_1}\mathcal{L}_{l_1,n_1+1}^{a_0,a}(p_1,\cdots,p_{n_1})\left(\partial^{w_3}\partial_{1/a}C^{a_0,a}(\hat{p})-\partial^{w_3}\partial_{1/a}C^{a_0,a}(\hat{p}^O)\right)\partial^{w_2}\mathcal{L}_{l_1,n_1+1}^{a_0,a,O}(-Op,\cdots,Op_{n})
\end{eqnarray}{}
Our bound on $\mathcal{D}_{l,n}^{a_0,a}$ will be verified by proving it for these difference terms. 
\begin{itemize}
    \item We first bound (\ref{b1}). Using inequality (\ref{bou}) for $\partial^w 
    \mathcal{L}^{a_0,a,O}_{l,n}(Op_1,\cdots,Op_n)$ which can be proven  as it was shown for $\partial^w\mathcal{L}^{a_0,a}_{l,n}$, we obtain 
\begin{eqnarray*}
    \fl \left | \int_{k,\mathcal{B}_{a_0}}\partial^w \mathcal{L}_{l-1,n+1}^{a_0,a,O}(Ok,Op_1,\cdots,Op_n,-Ok)
    \left[\partial_{1/a}C^{a_0,a}(\hat{k})-\partial_{1/a}C^{a_0,a}(\hat{k}^O)\right]\right |\\
    \fl \leq \int_{k,\mathcal{B}_{a_0}}2a^3\left(\frac{1}{a}+m\right)^{3-n-|w|}\left | e^{-a^2(\hat{k}^2+m^2)}-e^{-a^2((\hat{k}^O)^2+m^2)}\right|\\
    \indent \indent \qquad \mathcal{P}_1\left(\log \frac{1+am}{am}\right)\mathcal{P}_2\left(\frac{a|k|}{1+am},\left \{\frac{a|p_i|}{1+am}\right \}\right)\nonumber
\end{eqnarray*}
We define $$\mathcal{I}^{O}_{a_0}:=\left \{ k \in \mathcal{B}_{a_0}: \sum_{\mu=1}^4 
\sin^2 \frac{a_0k_{\mu}}{2}\leq \sum_{\mu=1}^4 \sin^2 \frac{a_0(Ok)_{\mu}}{2} \right \}$$
We decompose the integral over the Brillouin zone $\mathcal{B}_{a_0}$ 
into integrals over $\mathcal{I}_{a_0}^O$ and ${\mathcal{I}_{a_0}^O}^c$,
\begin{eqnarray*}
   \fl \eqalign{\int_{k,\mathcal{B}_{a_0}}} 2a^3
   \left | e^{-a^2(\hat{k}^2+m^2)}-e^{-a^2{ (\hat{k}^O)^2+m^2)}} \right| \mathcal{P}\left(\frac{a|k|}{1+am},\left \{\frac{a|p_i|}{1+am}\right \}\right)\\
   = \int_{k,\mathcal{I}_{a_0}^O}2a^3e^{-a^2(\hat{k}^2+m^2)}\left | e^{-a^2((\hat{k}^O)^2-\hat{k}^2)}-1\right|\mathcal{P}\left(\frac{a|k|}{1+am},\left \{\frac{a|p_i|}{1+am}\right \}\right)\\
  +\int_{(k,\mathcal{I}_{a_0}^O)^c}2a^3e^{-a^2((\hat{k}^O)^2+m^2)}\
\left | e^{-a^2(\hat{k}^2-(\hat{k}^O)^2)}-1\right| \mathcal{P}\left(\frac{a |k|}{1+am},\left \{\frac{a|p_i|}{1+am}\right \}\right)
\end{eqnarray*}
From the definition of $\mathcal{I}_{a_0}^O$, we have
$$\forall k \in \mathcal{I}_{a_0}^O,~ |\hat{k}|\leq |\hat{k}^O|~~~~~~\forall k \in {\mathcal{I}^{O}_{a_0}}^c,~|\hat{k}|>|\hat{k}^O|$$
which implies that 
$$\left | e^{-a^2( (\hat{k}^O)^2-\hat{k}^2)}-1\right|\leq a^2\left| (\hat{k}^O)^2-\hat{k}^2 \right|,~~~\forall k \in \mathcal{I}_{a_0}^O$$
$$\left | e^{-a^2(\hat{k}^2-(\hat{k}^O)^2)}-1\right|\leq a^2\left| (\hat{k}^O)^2-\hat{k}^2 \right|,~~~\forall k \in {\mathcal{I}_{a_0}^O}^c$$
Using the bound (\ref{prp}), we obtain
\begin{eqnarray}{}\label{wzeron}
a^2\left| (\hat{k}^O)^2-\hat{k}^2\right|\leq \frac{a_0}{a}\mathcal{P}\left(a|k|\right)
\end{eqnarray}

This gives the following bound    
\[
\int_{k,\mathcal{I}_{a_0}^O}a^3e^{-a^2(\hat{k}^2+m^2)}\left | e^{-a^2((\hat{k}^O)^2-\hat{k}^2)}-1\right|\ \mathcal{P}\left(\frac{a|k|}{1+am},\left \{\frac{a|p_i|}{1+am}\right \}\right)
\]
\begin{equation}
\leq \frac{a_0}{a}\int_{k,\mathcal{B}_{a_0}}a^3e^{-a^2(\hat{k}^2+m^2)}\ \mathcal{P}\left({a|k|},\left \{\frac{a|p_i|}{1+am}\right \}\right)
    \leq \frac{a_0}{a}\ \tilde{\mathcal{P}}\left(\left \{\frac{a|p_i|}{1+am}\right \}\right)
\label{2}
\end{equation}
where the last inequality follows from  lemma \ref{lemma1}.
Similarly, we obtain for the second integral depending over ${\mathcal{I}^O_{a_0}}^c$
\begin{eqnarray*}
\fl \int_{k,{\mathcal{I}_{a_0}^O}^c}a^3e^{-a^2((\hat{k}^O)^2+m^2)}
\left | e^{-a^2(\hat{k}^2-(\hat{k}^O)^2)}-1\right|\mathcal{P}\left(\frac{a|k|}{1+am},\left \{\frac{a|p_i|}{1+am}\right \}\right) \\
\leq \frac{a_0}{a}\int_{k,{\mathcal{I}_{a_0}^O}^c}a^3e^{-a^2((\hat{k}^O)^2+m^2)}\ \mathcal{P}\left({a|k|},\left \{\frac{a|p_i|}{1+am}\right \}\right)
\end{eqnarray*}{}
Performing the change of variables $k\rightarrow Ok$ yields
\begin{eqnarray}
\fl  \int_{k,{\mathcal{I}_{a_0}^O}^c}a^3e^{-a^2((\hat{k}^O)^2+m^2)}\ \mathcal{P}\left(a|k|,\left \{\frac{a|p_i|}{1+am}\right \}\right)\nonumber\\
   =\int_{k,O{(\mathcal{I}_{a_0}^O)^c}}a^3e^{-a^2({\hat{k}}^2+m^2)}\ \mathcal{P}\left({a|O^{-1}k|},\left \{\frac{a|p_i|}{1+am}\right \}\right)\nonumber\\
  \leq \int_{k,\mathcal{B}_{\alpha a_0}}a^3e^{-a^2({\hat{k}}^2+m^2)}\ \mathcal{P}\left({a|k|},\left \{\frac{a|p_i|}{1+am}\right \}\right)
\leq\ \tilde{\mathcal{P}}\left(\left \{\frac{a|p_i|}{1+am}\right \}\right)
\label{3}
\end{eqnarray}
where $\alpha$ is a parameter strictly less than $1$ such that $O\mathcal{B}_{a_0}\subset \mathcal{B}_{\alpha a_0}$,
and the last inequality follows again from lemma \ref{lemma1}.
Combining (\ref{2}) and (\ref{3}) the first difference term is bounded
\begin{eqnarray}\label{b11}
   \fl \left | \int_{k,\mathcal{B}_{a_0}}\partial^w \mathcal{L}_{l-1,n+1}^{a_0,a,O}(Ok,Op_1,\cdots,Op_n,-Ok)\left[\partial_{1/a}C^{a_0,a}(\hat{k})-\partial_{1/a}C^{a_0,a}(\hat{k}^O)\right]\right |\nonumber\\
    \leq a_0\left(\frac{1}{a}+m\right)^{4-n-|w|}
     \mathcal{P}_1\left(\log \frac{1+a_0m}{a_0m}\right)\mathcal{P}_2\left(\left \{\frac{a|p_i|}{1+am}\right \}\right)
\end{eqnarray}{}






\item The second step is to bound (\ref{b2}). For this step we use  lemma \ref{lemma2}.

  Using (\ref{bou}) for $\partial^{w_1} \mathcal{L}^{a_0,a}_{l_1,n_1+1}$ and $\partial^{w_2}\mathcal{L}^{a_0,a,O}_{l_2,n_2+1}$ we obtain
\begin{eqnarray*}
    \fl \left|{\partial}^{w_1}\mathcal{L}_{l_1,n_1+1}^{a_0,a}
    \left({\partial}^{w_3}\partial_{1/a}C^{a_0,a}(\hat{p})-{\partial}^{w_3}\partial_{1/a}C^{a_0,a}(\hat{p}^O)\right){\partial}^{w_2}
    \mathcal{L}_{l_2,n_2+1}^{a_0,a,O}\right|\\
    \leq a_0 \left(\frac{1}{a}+m\right)^{4-n-|w|}\mathcal{P}_1\left(\log \frac{1+a_0m}{a_0m}\right)\mathcal{P}_2\left(\left \{\frac{a|p_i|}{1+am}\right \}\right)
\end{eqnarray*}
Using the induction bound on $\partial^{w_i}\mathcal{D}_{l_i,n_i+1}^{a_0,a}$ and the bound (\ref{propder}), we deduce that
\begin{eqnarray*}
   \fl \left|{\partial}^{w_j}\mathcal{L}_{l_j,n_j+1}^{a_0,a}{\partial}^{w_3}\partial_{1/a}C^{a_0,a}(\hat{p}){\partial}^{w_i}\mathcal{D}_{l_i,n_i+1}^{a_0,a}\right|\\
    \leq a_0 \left(\frac{1}{a}+m\right)^{4-n-|w|}\mathcal{P}_1\left(\log \frac{1+a_0m}{a_0m}\right)\mathcal{P}_2\left(\left \{\frac{a|p_i|}{1+am}\right \}\right)
\end{eqnarray*}
Combining all the previous estimates of each term of the r.h.s. of the FE (\ref{Dln}), we obtain
\begin{equation}\label{PartDln}
    \fl \left|\partial_{1/a}\partial^w\mathcal{D}_{l,n}^{a_0,a}(p_1,\cdots,p_n)\right|\leq a_0 
    \left(\frac{1}{a}+m\right)^{4-n-|w|}\mathcal{P}_1\left(\log \frac{1+a_0m}{a_0m}\right)\mathcal{P}_2\left(\left \{\frac{a|p_i|}{1+am}\right \}\right)
\end{equation}{}

\item After these preparation steps, we integrate the flow equations (\ref{Dln}):
\begin{itemize}
    \item [C1)] For the irrelevant terms, because of the boundary conditions
    \begin{eqnarray*}{}
    \eqalign{\partial^w \mathcal{D}_{l,n}^{a_0,a_0}(p_1,\cdots,p_n)&=0,\qquad \forall n+|w|\geq 5~ (n\neq 2)\\
    \partial^w \mathcal{D}_{l,2}^{a_0,a_0}(p,-p)&=b_l(a_0)\partial^w((\hat{p}^O)^2-\hat{p}^2),\qquad \forall |w|\geq 3}
    \end{eqnarray*}
we integrate from $1/a_0$ to $1/a$. We exclude for the moment $(n,|w|)\in \left \{ \right (4,1);(2,3)\}$ which have to be treated as relevant in this case.\\
$\forall n+|w|>5$, such that $n\neq 2$ we have
\begin{eqnarray*}
    \fl \left | {\partial}^{w}\mathcal{D}_{l,n}^{a_0,a}\left({p}_1,\cdots,{p}_n\right) \right | \\ \fl \qquad \leq \int_{1/a}^{1/a_0}d\lambda \left | \partial_{\lambda}{\partial}^{w}\mathcal{D}_{l,n}^{a_0,\frac{1}{\lambda}}\left({p}_1,\cdots,{p}_n\right) \right |\\
    \fl  \qquad \leq a_0 \mathcal{P}_1\left(\log \frac{1+a_0m}{a_0m}\right)\mathcal{P}_2\left(\left \{\frac{a|p_i|}{1+am}\right \}\right)\int_{1/a}^{1/a_0} d\lambda \left( \lambda+m\right)^{5-n-|w|-1}\\
    \fl \qquad \leq a_0 \left( \frac{1}{a}+m \right)^{5-n-|w|} \mathcal{P}_1\left( \log \frac{1+a_0m}{a_0m}\right) \mathcal{P}_2\left( \left \{\frac{a|p_i|}{1+am}\right \}\right)
\end{eqnarray*}
For $n=2$ and $|w|\geq 4$, the boundary conditions are not equal to zero. Therefore,
\begin{eqnarray*}
    \fl \left | {\partial}^{w}\mathcal{D}_{l,2}^{a_0,a}\left({p},-{p}\right) \right | &\leq \int_{1/a}^{1/a_0}d\lambda \left | \partial_{\lambda}{\partial}^{w}\mathcal{D}_{l,2}^{a_0,\frac{1}{\lambda}}\left({p},-{p}\right) \right |+\left | {\partial}^{w}\mathcal{D}_{l,2}^{a_0,a_0}\left({p},-{p}\right) \right |
\end{eqnarray*}
We recall that
$${\partial}^{w}\mathcal{D}_{l,2}^{a_0,a_0}\left({p},-{p}\right)=b_l(a_0)\partial^w\left((\hat{p}^O)^2-\hat{p}^2\right)$$
Due to (\ref{bla0}) 
\begin{equation}
\partial_{p^2}\mathcal{L}_{l,2}^{a_0,a_0}(0,0)=2b_l(a_0) \ \leq \ \mathcal{P}\left(\log\frac{1+a_0m}{a_0m}\right)
\label{DL22}
\end{equation}
(\ref{dar}), (\ref{t3}) together with (\ref{DL22}) imply 
\begin{equation}
    \fl \left | {\partial}^{w}\mathcal{D}_{l,2}^{a_0,a}\left({p},-{p}\right) \right | \leq a_0 \left( \frac{1}{a}+m \right)^{3-n-|w|} \mathcal{P}_1\left( \log \frac{1+a_0m}{a_0m}\right) \mathcal{P}_2\left( \frac{a |p|}{1+am}\right)
\end{equation}
\item [C2)] For the cases $n+|w|\leq 5$, the claim (\ref{Dlnine1}) has to be deduced from the respective integrated flow equation (\ref{Dln}) at the renormalization point followed by an extension to general momenta with the aid of the  Taylor Formula (\ref{taylorfor}) applied to $\mathcal{D}_{l,2}^{a_0,a}$.
We proceed in the order of the induction starting with the cases $(n=2,|w|=3)$, $(n=2,|w|=2)$ and going down in $|w|$. 
The integral in
\begin{equation}
\fl {\partial}^{w}\mathcal{D}_{l,n}^{a_0,a}\left(0,\cdots,0\right)={\partial}^{w}\mathcal{D}_{l,n}^{a_0,\infty}\left(0,\cdots,0\right)+\displaystyle \int_0^{1/a}d\lambda{\partial}^{w}\mathcal{D}_{l,n}^{a_0,1/\lambda}\left(0,\cdots,0\right)
\end{equation}
is bounded using  (\ref{PartDln}) at vanishing momenta:
\begin{eqnarray*}
\fl \left | \displaystyle \int_0^{1/a}d\lambda{\partial}^{w}\mathcal{D}_{l,n}^{a_0,1/\lambda}\left(0,\cdots,0\right)\right| & \leq a_0 \displaystyle \int_0^{1/a}d\lambda\left(\lambda+m\right)^{5-n-|w|-1}\mathcal{P}\left(\log\frac{1+a_0m}{a_0m}\right)\\
    & \leq a_0 \left(\frac{1}{a}+m\right)^{5-n-|w|}\mathcal{P}\left(\log\frac{1+a_0m}{a_0m}\right)\\
\end{eqnarray*}
Hence, the assertion is established at the renormalization point. In each case extension to general momenta via (\ref{taylorfor}) is guaranteed
 by bounds established before. This concludes the proof of Theorem \ref{O(4)-sym}.
\end{itemize}{}
\end{itemize}
\end{proof}


\section{Translation invariance}\label{Sec5}
\subsection{Some properties of the Schwartz space}
We recall the definition of the Schwartz space 
$$\mathcal{S}\left(\mathbb{R}^{4n}\right):=\left\{ f \in \mathcal{C}^{\infty}\left(\mathbb{R}^{4n}\right)\left| \right. \forall\left(\alpha,\beta\right)\in \mathbb{N}^{4n}\times \mathbb{N}^{4n},\sup_{x \in \mathbb{R}^{4n}}\left|x^{\alpha}D^{\beta}f(x)\right|<+\infty\right\}$$
The Schwartz space is a Fréchet space endowed with a topology induced by the filtrant family of semi-norms
$$\mathcal{N}_p\left(\cdot\right)=\sum_{|\alpha|,|\beta|\leq p}\left \| \cdot\right \|_{\alpha,\beta},~~p\in \mathbb{N}$$
where 
$$\left \| f\right \|_{\alpha,\beta}:=\sup_{x \in \mathbb{R}^{4n}}\left|x^{\alpha}D^{\beta}f(x)\right|$$
\begin{lemma}\label{lemmeS}
Let $f \in \mathcal{S}\left(\mathbb{R}^{4n}\right)$ and $\mathcal{P}_r$ a polynomial of degree $r$, we have the following bound 
\begin{eqnarray*}
\left| \mathcal{P}_r\left(x_1,\cdots,x_n\right)f(x_1,\cdots,x_n)\right|\leq \left(\prod_{i=1}^n \frac{1}{\left(1+|x_i|\right)^s}\right)\mathcal{N}_{s+r}\left(f\right),\indent \forall s \in \mathbb{N}
\end{eqnarray*}
\end{lemma}
\noindent The proof of Lemma \ref{lemmeS} which we do not reproduce here uses the definition of Schwartz functions and will be useful in the sequel. For more details about the properties of Schwartz space and tempered distributions, we refer the reader to \cite{22}.
\subsection{Translation invariance}
The lattice breaks Euclidean translation invariance. In this section, we prove that the continuum limit restores translation invariance. \\
The regularized (CAS) n-point functions in position space are tempered distributions that we define by their Fourier transform, that is for $f \in \mathcal{S}(\mathbb{R}^{4n})$
\begin{eqnarray*}
\fl \langle\mathcal{L}_{l,n, \Lambda_{a_0}}^{a_0,a},f\rangle_{\mathcal{S}^{'},S} &:=\int_{\mathcal{B}_{a_0}^n} \frac{d^4p_1 \cdots d^4p_n}{(2\pi)^{4n}}\, \mathcal{L}_{l,n}^{a_0,a}(p_1,\cdots,p_n)\, \delta^{(4)}_{\left[\frac{2\pi}{a_0}\right]}( p_1+\cdots+p_n)\,\mathcal{F}^{-1}(f)\left(p_1,\cdots,p_n\right)
\end{eqnarray*}{}
where
$$\delta^{(4)}_{\left[\frac{2\pi}{a_0}\right]}(p_1+\cdots+p_n):=\sum_{k \in \mathbb{Z}^4}\delta^{(4)}\left(p_1+\cdots+p_n-{\frac{2k\pi}{a_0}}\right)$$
accounts for the invariance of $\mathcal{L}_{l,n}^{a_0,a}$ under lattice translations and $\mathcal{F}^{-1}(f)$ is the inverse Fourier transform of $f$. $\mathcal{L}^{a_0,a}_{l,n,\Lambda_{a_0}}$ is well defined as a tempered distribution since $$\mathcal{L}^{a_0,a}_{l,n}\left(p_1,\cdots,p_n\right)\delta^{(4)}_{\left[\frac{2\pi}{a_0}\right]}( p_1+\cdots+p_n)$$
is a $\frac{2\pi}{a_0}$-periodic distribution \cite{22}.
\\
Similarly, we define the renormalized (CAS) n-point functions in the position space
\begin{eqnarray*}
\fl \indent \langle\mathcal{L}_{l,n,x}^{0,\infty},f\rangle_{\mathcal{S}^{'},S}&:=\int_{\mathbb{R}^{4n}} \frac{d^4p_1 \cdots d^4p_n}{(2\pi)^{4n}}\, \mathcal{L}_{l,n}^{0,\infty}(p_1,\cdots,p_n)\, \delta^{(4)}( p_1+\cdots+p_n)\,\mathcal{F}^{-1}(f)\left(p_1,\cdots,p_n\right)
\end{eqnarray*}
 $\mathcal{L}^{0,\infty}_{l,n,x}$ denotes the continuum limit position space (CAS) n-point function. It is a tempered distribution for which the translation by a vector $c\in \mathbb{R}^4$ is defined as 
 $$\langle  \tau_c \mathcal{L}_{l,n,x}^{0,\infty},f\rangle_{\mathcal{S}^{'},S}:=\langle \mathcal{L}_{l,n,x}^{0,\infty},\tau_{-c}f\rangle_{\mathcal{S}^{'},S},\indent \forall f \in \mathcal{S}\left(\mathbb{R}^{4n}\right)$$
 and $$\left(\tau_{-c}f\right)(p_1,\cdots,p_n):=f\left(p_1+c,\cdots,p_n+c\right)$$
 Therefore,
 \begin{eqnarray*}
\fl \langle\tau_{c}\mathcal{L}_{l,n,x}^{0,\infty},f\rangle_{\mathcal{S}^{'},S} \\ \fl \indent =\int_{\mathbb{R}^{4n}} \prod_{i=1}^n\frac{d^4p_i}{(2\pi)^{4n}}\, \mathcal{L}_{l,n}^{0,\infty}(p_1,\cdots,p_n)\, \delta^{(4)}( p_1+\cdots+p_n)\,e^{-i(p_1+\cdots+p_n)\cdot c}\mathcal{F}^{-1}(f)\left(p_1,\cdots,p_n\right)
\end{eqnarray*}
which implies 
$$\langle\tau_{c}\mathcal{L}_{l,n,x}^{0,\infty},f\rangle_{\mathcal{S}^{'},S}=\langle \mathcal{L}_{l,n,x}^{0,\infty},f\rangle_{\mathcal{S}^{'},S},\indent \forall f \in \mathcal{S}\left(\mathbb{R}^{4n}\right)$$
The continuum limit is clearly invariant under translations. Thus, proving the translation invariance of the continuum limit amounts to establishing the following convergence
 \begin{theorem}\label{trans}
Let $f \in \mathcal{S}\left(\mathbb{R}^{4n}\right)$, 
\begin{eqnarray}\label{con01}
    \langle \mathcal{L}_{l,n,\Lambda_{a_0}}^{a_0,a},f\rangle_{\mathcal{S}^{'},S} \longrightarrow \langle \mathcal{L}_{l,n,x}^{0,\infty},f\rangle_{\mathcal{S}^{'},S}\  \mbox{ for } \ a_0 \rightarrow 0, a\rightarrow \infty
\end{eqnarray}{}
\end{theorem}
\noindent The proof of Theorem \ref{trans} relies on the following lemma
\begin{lemma}\label{lemmaP}
Let $f \in \mathcal{S}\left(\mathbb{R}^{4n}\right)$, 
\begin{eqnarray}\label{con01}
    \langle \delta_{\left[\frac{2\pi}{a_0}\right]}^{(4)}(p_1+\cdots+p_n),f\rangle_{\mathcal{S}^{'},S} \longrightarrow \langle \delta^{(4)}(p_1+\cdots+p_n),f\rangle_{\mathcal{S}^{'},S}\  \mbox{ for } \ a_0 \rightarrow 0
\end{eqnarray}{}
\end{lemma}

\subsection{Proof of lemma \ref{lemmaP}}
\begin{proof}
Let $f \in \mathbb{R}^{4n}$, using lemma \ref{lemmeS}, one can verifies that 
$$\left| \langle \delta^{4}\left(p_1+\cdots+p_n\right),f\rangle_{\mathcal{S}^{'},\mathcal{S}}\right| \leq C \mathcal{N}_{5(n-1)}(f)$$
which proves that $\delta^{4}\left(p_1+\cdots+p_n\right)$ is a tempered distribution.
We have that 
\begin{eqnarray*}
   \fl \langle \delta^{(4)}_{\left[\frac{2\pi}{a_0}\right]},f\rangle_{\mathcal{S}^{'}, \mathcal{S}}=\int_{\mathbb{R}^{4n}}d^4p_1\cdots d^4p_n \delta^{(4)}_{\left[\frac{2\pi}{a_0}\right]}(p_1+\cdots+p_n)f(p_1+\cdots+p_n)\\
   =\int_{\mathbb{R}^{4(n-1)}}d^4p_1\cdots d^4p_n \sum_{k \in \mathbb{Z}^4}f\left(\frac{2k\pi}{a_0}-\sum_{i=2}^n p_i+\cdots+p_n\right)
\end{eqnarray*}
We write 
$$\sum_{k \in \mathbb{Z}^4}f\left(\frac{2k\pi}{a_0}-\sum_{i=2}^n p_i,\cdots,p_n\right)=f\left(-\sum_{i=2}^n p_i,\cdots,p_n\right)+\sum_{k \in \mathbb{Z}^{4,*}}f\left(\frac{2k\pi}{a_0}-\sum_{i=2}^n p_i,\cdots,p_n\right)$$
Since $f\in \mathcal{S}(\mathbb{R}^{4n})$, we have the following bound for any $k \in \mathbb{Z}^{4,*}$,
\begin{eqnarray*}
\fl \left|f\left(\frac{2k\pi}{a_0}-\sum_{i=2}^n p_i,p_2,\cdots,p_n\right)\right|\\
\leq \frac{1}{\left(\left|\sum_{i=2}^n p_i\right|^2+\left|\frac{2k\pi}{a_0}-\sum_{i=2}^n p_i\right|^2\right)^4}~~\prod_{i=2}^n\frac{1}{\left(1+|p_i|\right)^5}~~{N}_{s}(f)
\end{eqnarray*}
where $s=5(n-1)+8$ and  $${N}_{s}(f)=\sup_{p_i \in \mathbb{R}^4}\sup_{|\alpha|\leq 13n}|p_1|^{\alpha_1}\cdots |p_n|^{\alpha_{n}}\left|f(p_1+\cdots+p_n)\right|$$
Using $$\frac{1}{|a|^2+|b|^2}\leq \frac{2}{|a+b|^2} ,\indent \forall a,b \in \mathbb{R}^{p,*}$$
we obtain
\begin{eqnarray*}
\fl \left|f\left(\frac{2k\pi}{a_0}-\sum_{i=2}^n p_i,p_2,\cdots,p_n\right)\right|
\leq C \left(\frac{a_0}{|k|}\right)^{8}~~\prod_{i=2}^n\frac{1}{\left(1+|p_i|\right)^5}~~N_s(f)
\end{eqnarray*}
Using  
$$\sum_{k \in \mathbb{Z}^{4,*}}\frac{1}{|k|^8}\leq \sum_{k_i \in \mathbb{Z}^{*}}\prod_{i=1}^4\frac{1}{|k_i|^2}=\left(\frac{\pi^2}{3}\right)^4<+\infty$$
we obtain  
\begin{eqnarray*}
\fl \langle \delta^{(4)}_{\left[\frac{2\pi}{a_0}\right]},f\rangle_{\mathcal{S}^{'}, \mathcal{S}}&=\int_{\mathbb{R}^{4(n-1)}}d^4p_1\cdots d^4p_n f\left(-\sum_{i=2}^n p_i+\cdots+p_n\right)+C~a_0^8~N_s(f)\\
&=\langle \delta^{(4)}(p_1+\cdots+p_n),f\rangle_{\mathcal{S}^{'}, \mathcal{S}}+C~a_0^8~N_s(f)
\end{eqnarray*}
together with the useful bound
\begin{eqnarray}\label{useb}
\left| \langle \delta^{(4)}_{\left[\frac{2\pi}{a_0}\right]},f\rangle_{\mathcal{S}^{'}, \mathcal{S}}\right|&\leq C(1+a_0^8)N_s(f)
\end{eqnarray}
This proves that for $a_0\rightarrow 0$ we have
\begin{eqnarray*}
\langle \delta^{(4)}_{\left[\frac{2\pi}{a_0}\right]}(p_1+\cdots+p_n),f\rangle_{\mathcal{S}^{'}, \mathcal{S}}\rightarrow_{a_0\rightarrow0}\langle \delta^{(4)}(p_1+\cdots+p_n),f\rangle_{\mathcal{S}^{'}, \mathcal{S}}
\end{eqnarray*}
\end{proof}
\subsection{Proof of Theorem \ref{trans}}
\begin{proof}
We recall the boundedness inequality (\ref{bou}) for the (CAS) n-point functions. For all $(p_i)_{1\leq i\leq n}\in \mathcal{B}_{a_0}$ such that $\sum_{i=1}^n p_i \equiv 0 \left[\frac{2\pi}{a_0}\right]$, we have
\begin{eqnarray*}
\left | {\partial}^{w}\mathcal{L}_{l,n}^{a_0,a}({p}_{1},\cdots,{p}_{n}) \right | \leq \left( \frac{1}{a}+m \right)^{4-n-|w|} \mathcal{P}_1\left( \log \frac{1+am}{am} \right) \mathcal{P}_2 \left( \left \{ \frac{a|{p}_i|}{1+am} \right \} \right)
\end{eqnarray*}
This proves that $\mathcal{L}^{a_0,a}_{l,n}$ are $\mathcal{C}^{\infty}$ w.r.t. to the momenta and are at most of polynomial growth. Therefore, $$\forall f \in \mathcal{S}\left(\mathbb{R}^{4n}\right),\indent \mathbb{1}_{\mathcal{B}_{a_0}}(p_1,\cdots,p_n) \mathcal{L}^{a_0,a}_{l,n}(p_1,\cdots,p_n)f \in \mathcal{S}\left(\mathbb{R}^{4n}\right)$$
Taking the limit in the boundedness inequality (\ref{bou}), the same reasoning applies to $\mathcal{L}^{0,\infty}_{l,n}(p_1,\cdots,p_n)$ to prove that
$$\forall f \in \mathcal{S}\left(\mathbb{R}^{4n}\right),\indent  \mathcal{L}^{0,\infty}_{l,n}(p_1,\cdots,p_n)f \in \mathcal{S}\left(\mathbb{R}^{4n}\right)$$
We write 
\begin{eqnarray*}
    g_{a_0,a}(p_1,\cdots,p_n)&:=&\mathbb{1}_{\mathcal{B}_{a_0}}(p_1,\cdots,p_n) \mathcal{L}^{a_0,a}_{l,n}(p_1,\cdots,p_n)f(p_1,\cdots,p_n)\\
    g(p_1,\cdots,p_n)&:=& \mathcal{L}^{0,\infty}_{l,n}(p_1,\cdots,p_n)f(p_1,\cdots,p_n)
\end{eqnarray*}
Using (\ref{useb}), we obtain
$$\left| \langle \delta^{(4)}_{\left[\frac{2\pi}{a_0}\right]}, g_{a_0,a}-g\rangle \right|\leq C (1+a_0^8)~{N}_{s}(g_{a_0,a}-g)$$
Taking the limit ${\hat a}_0\rightarrow 0$ in (\ref{cauchy}) we find
\begin{eqnarray*}
     \left| \mathcal{L}_{l,n}^{a_0,\infty}(p_1,\cdots,p_n)-\mathcal{L}_{l,n}^{0,\infty}(p_1,\cdots,p_n)\right|
    \leq a_0 m^{5-n} \left(\log \frac{1}{a_0m}\right)^{\nu}\mathcal{P}\left(\left\{\frac{|p_i|}{m}\right \}\right)
\end{eqnarray*}{}
where $\nu$ is the same constant of (\ref{cauchy}).
Therefore, for any polynomial $\mathcal{Q}$ with nonnegative coefficients we obtain 
\begin{eqnarray*}
    \fl \left|\mathcal{Q}\left(\left \{{|p_i|}\right \}\right)\left( \mathcal{L}_{l,n}^{a_0,\infty}(p_1,\cdots,p_n)-\mathcal{L}_{l,n}^{0,\infty}(p_1,\cdots,p_n)\right)f\right|
    \leq a_0 m^{5-n} \left(\log \frac{1}{a_0m}\right)^{\nu}\tilde{\mathcal{P}}\left(\left\{\frac{|p_i|}{m}\right \}\right)|f(p_1,\cdots,p_n)|
\end{eqnarray*}{}
Thus, 
$$N_s(g_{a_0,a}-g)\leq a_0 m^{5-n} \left(\log \frac{1}{a_0m}\right)^{\nu}N_r(f)$$
which implies 
$$\langle \delta^{(4)}_{\left[\frac{2\pi}{a_0}\right]}, g_{a_0,a}-g\rangle \rightarrow_{a_0\rightarrow 0, a\rightarrow \infty}0$$
Lemma \ref{lemmaP} gives that 
$$\langle \delta^{(4)}_{\left[\frac{2\pi}{a_0}\right]}-\delta^{(4)}, g\rangle \rightarrow_{a_0\rightarrow 0, a\rightarrow \infty}0$$
so that 
$$\langle \delta^{(4)}_{\left[\frac{2\pi}{a_0}\right]}, g_{a_0,a}\rangle \rightarrow_{a_0\rightarrow 0, a\rightarrow \infty}\langle \delta^{(4)}, g\rangle $$
that is for all $f \in \mathcal{S}(\mathbb{R}^{4n})$,
\begin{eqnarray*}
    \langle \mathcal{L}_{l,n,\Lambda_{a_0}}^{a_0,a},f\rangle_{\mathcal{S}^{'},S} \longrightarrow \langle \mathcal{L}_{l,n,x}^{0,\infty},f\rangle_{\mathcal{S}^{'},S}\  \mbox{ for } \ a_0 \rightarrow 0, a\rightarrow \infty
\end{eqnarray*}{}
\end{proof}


\section*{Concluding remarks}
\indent We have presented an alternative proof of the perturbative renormalizability of massive lattice regularized $\phi_4^4$-theory. 
The starting point were the bounds (\ref{bou})-(\ref{con}) which prove the existence of the continuum limit. In the flow equation formalism, 
they serve at the same time as induction hypotheses for the inductive proof. Bounds of this sort have been established rigorously for all 
theories of physical interest, including gauge theories \cite{21}.

In this context  it is  also interesting to study the difference 
$$\mathcal{L}_{l,n,\Lambda_0}^{a_0,a}-\mathcal{L}_{l,n,a_0}^{a_0,a}$$
where $\mathcal{L}_{l,n,\Lambda_0}^{a_0,a}$ denotes the momentum space regularized correlation functions and $\mathcal{L}_{l,n,a_0}^{a_0,a}$ denotes the 
lattice regularized correlation functions. The UV-cutoff can be related to the lattice parameter by $\Lambda_0=1/a_0$,  similarly for the 
corresponding flowing parameters $\Lambda=1/a$. The study of this difference by flow equations should allow to prove that in the limit 
$a_0\rightarrow0$ and $a\rightarrow\infty$, the difference vanishes, implying consistency, that is the two regularization schemes converge to the same limit. 
This would be an alternative way to prove that the continuum limit when the lattice regularization is removed yields $O(4)$-symmetric correlation functions.

 We are confident  that our approach could be generalized to massless lattice regularized theories. In this case the appearing infrared singularities have 
to be controlled in a similar  way as it has been done for theories with momentum cutoff regularization \cite{21}. A particularly interesting subject is the extension
to  gauge theories since the lattice regularization respects a priori gauge invariance. It seems however that analyzing the flow equations still requires a gauge
fixing procedure. In any case the important issue is to prove that the continuum limit respects the continuum  Ward identities  for suitable renormalization
conditions.

\end{document}